\documentclass[
	10pt, 
	two column, 
	twoside
]{IEEEtran}
\usepackage{fancyhdr}
\usepackage{amsmath,epsfig}
\usepackage{threeparttable}
\usepackage{epsf,epsfig}
\usepackage{amsthm}
\usepackage{amsmath}
\usepackage{amssymb}
\usepackage{amsfonts}
\usepackage[ruled]{algorithm2e}
\usepackage[noadjust]{cite}
\usepackage{dsfont}
\usepackage{bbding}
\usepackage{pifont}
\usepackage{framed}
\usepackage{soul}
\usepackage{color}
\usepackage{subfigure}
\usepackage{caption}
\usepackage{subcaption}
\usepackage{url}
\usepackage{orcidlink}

\begin{document}
\newtheorem{theorem}{Theorem}
\newtheorem{acknowledgement}[theorem]{Acknowledgement}
\newtheorem{axiom}[theorem]{Axiom}
\newtheorem{case}[theorem]{Case}
\newtheorem{claim}[theorem]{Claim}
\newtheorem{conclusion}[theorem]{Conclusion}
\newtheorem{condition}[theorem]{Condition}
\newtheorem{conjecture}[theorem]{Conjecture}
\newtheorem{criterion}[theorem]{Criterion}
\newtheorem{definition}{Definition}
\newtheorem{exercise}[theorem]{Exercise}
\newtheorem{lemma}{Lemma}
\newtheorem{corollary}{Corollary}
\newtheorem{notation}[theorem]{Notation}
\newtheorem{problem}[theorem]{Problem}
\newtheorem{proposition}{Proposition}
\newtheorem{solution}[theorem]{Solution}
\newtheorem{summary}[theorem]{Summary}
\newtheorem{assumption}{Assumption}
\newtheorem{example}{\bf Example}
\newtheorem{remark}{\bf Remark}

\def\qed{$\Box$}
\def\QED{\mbox{\phantom{m}}\nolinebreak\hfill$\,\Box$}
\def\proof{\noindent{\emph{Proof:} }}
\def\poof{\noindent{\emph{Sketch of Proof:} }}
\def
\endproof{\hspace*{\fill}~\qed
\par
\endtrivlist\unskip}
\def\endproof{\hspace*{\fill}~\qed\par\endtrivlist\vskip3pt}

\def\E{\mathsf{E}}
\def\eps{\varepsilon}
\def\phi{\varphi}
\def\Lsp{{\boldsymbol L}}
\def\Bsp{{\boldsymbol B}}
\def\lsp{{\boldsymbol\ell}}
\def\Ltsp{{\Lsp^2}}
\def\Lpsp{{\Lsp^p}}
\def\Linsp{{\Lsp^{\infty}}}
\def\LtR{{\Lsp^2(\Rst)}}
\def\ltZ{{\lsp^2(\Zst)}}
\def\ltsp{{\lsp^2}}
\def\ltZt{{\lsp^2(\Zst^{2})}}
\def\ninN{{n{\in}\Nst}}
\def\oh{{\frac{1}{2}}}
\def\grass{{\cal G}}
\def\ord{{\cal O}}
\def\dist{{d_G}}
\def\conj#1{{\overline#1}}
\def\ntoinf{{n \rightarrow \infty }}
\def\toinf{{\rightarrow \infty }}
\def\tozero{{\rightarrow 0 }}
\def\trace{{\operatorname{trace}}}
\def\ord{{\cal O}}
\def\UU{{\cal U}}
\def\rank{{\operatorname{rank}}}
\def\acos{{\operatorname{acos}}}

\def\SINR{\mathsf{SINR}}
\def\SNR{\mathsf{SNR}}
\def\SIR{\mathsf{SIR}}
\def\tSIR{\widetilde{\mathsf{SIR}}}
\def\Ei{\mathsf{Ei}}
\def\l{\left}
\def\r{\right}
\def\({\left(}
\def\){\right)}
\def\lb{\left\{}
\def\rb{\right\}}

\setcounter{page}{1}

\newcommand{\eref}[1]{(\ref{#1})}
\newcommand{\fig}[1]{Fig.\ \ref{#1}}

\def\bydef{:=}
\def\ba{{\mathbf{a}}}
\def\bb{{\mathbf{b}}}
\def\bc{{\mathbf{c}}}
\def\bd{{\mathbf{d}}}
\def\bee{{\mathbf{e}}}
\def\bff{{\mathbf{f}}}
\def\bg{{\mathbf{g}}}
\def\bh{{\mathbf{h}}}
\def\bi{{\mathbf{i}}}
\def\bj{{\mathbf{j}}}
\def\bk{{\mathbf{k}}}
\def\bl{{\mathbf{l}}}
\def\bn{{\mathbf{n}}}
\def\bo{{\mathbf{o}}}
\def\bp{{\mathbf{p}}}
\def\bq{{\mathbf{q}}}
\def\br{{\mathbf{r}}}
\def\bs{{\mathbf{s}}}
\def\bt{{\mathbf{t}}}
\def\bu{{\mathbf{u}}}
\def\bv{{\mathbf{v}}}
\def\bw{{\mathbf{w}}}
\def\bx{{\mathbf{x}}}
\def\by{{\mathbf{y}}}
\def\bz{{\mathbf{z}}}
\def\b0{{\mathbf{0}}}

\def\bA{{\mathbf{A}}}
\def\bB{{\mathbf{B}}}
\def\bC{{\mathbf{C}}}
\def\bD{{\mathbf{D}}}
\def\bE{{\mathbf{E}}}
\def\bF{{\mathbf{F}}}
\def\bG{{\mathbf{G}}}
\def\bH{{\mathbf{H}}}
\def\bI{{\mathbf{I}}}
\def\bJ{{\mathbf{J}}}
\def\bK{{\mathbf{K}}}
\def\bL{{\mathbf{L}}}
\def\bM{{\mathbf{M}}}
\def\bN{{\mathbf{N}}}
\def\bO{{\mathbf{O}}}
\def\bP{{\mathbf{P}}}
\def\bQ{{\mathbf{Q}}}
\def\bR{{\mathbf{R}}}
\def\bS{{\mathbf{S}}}
\def\bT{{\mathbf{T}}}
\def\bU{{\mathbf{U}}}
\def\bV{{\mathbf{V}}}
\def\bW{{\mathbf{W}}}
\def\bX{{\mathbf{X}}}
\def\bY{{\mathbf{Y}}}
\def\bZ{{\mathbf{Z}}}

\def\mA{{\mathbb{A}}}
\def\mB{{\mathbb{B}}}
\def\mC{{\mathbb{C}}}
\def\mD{{\mathbb{D}}}
\def\mE{{\mathbb{E}}}
\def\mF{{\mathbb{F}}}
\def\mG{{\mathbb{G}}}
\def\mH{{\mathbb{H}}}
\def\mI{{\mathbb{I}}}
\def\mJ{{\mathbb{J}}}
\def\mK{{\mathbb{K}}}
\def\mL{{\mathbb{L}}}
\def\mM{{\mathbb{M}}}
\def\mN{{\mathbb{N}}}
\def\mO{{\mathbb{O}}}
\def\mP{{\mathbb{P}}}
\def\mQ{{\mathbb{Q}}}
\def\mR{{\mathbb{R}}}
\def\mS{{\mathbb{S}}}
\def\mT{{\mathbb{T}}}
\def\mU{{\mathbb{U}}}
\def\mV{{\mathbb{V}}}
\def\mW{{\mathbb{W}}}
\def\mX{{\mathbb{X}}}
\def\mY{{\mathbb{Y}}}
\def\mZ{{\mathbb{Z}}}

\def\cA{\mathcal{A}}
\def\cB{\mathcal{B}}
\def\cC{\mathcal{C}}
\def\cD{\mathcal{D}}
\def\cE{\mathcal{E}}
\def\cF{\mathcal{F}}
\def\cG{\mathcal{G}}
\def\cH{\mathcal{H}}
\def\cI{\mathcal{I}}
\def\cJ{\mathcal{J}}
\def\cK{\mathcal{K}}
\def\cL{\mathcal{L}}
\def\cM{\mathcal{M}}
\def\cN{\mathcal{N}}
\def\cO{\mathcal{O}}
\def\cP{\mathcal{P}}
\def\cQ{\mathcal{Q}}
\def\cR{\mathcal{R}}
\def\cS{\mathcal{S}}
\def\cT{\mathcal{T}}
\def\cU{\mathcal{U}}
\def\cV{\mathcal{V}}
\def\cW{\mathcal{W}}
\def\cX{\mathcal{X}}
\def\cY{\mathcal{Y}}
\def\cZ{\mathcal{Z}}
\def\cd{\mathcal{d}}
\def\Mt{M_{t}}
\def\Mr{M_{r}}
\def\O{\Omega_{M_{t}}}
\newcommand{\figref}[1]{{Fig.}~\ref{#1}}
\newcommand{\tabref}[1]{{Table}~\ref{#1}}

\newcommand{\var}{\mathsf{var}}
\newcommand{\fb}{\tx{fb}}
\newcommand{\nf}{\tx{nf}}
\newcommand{\BC}{\tx{(bc)}}
\newcommand{\MAC}{\tx{(mac)}}
\newcommand{\Pout}{p_{\mathsf{out}}}
\newcommand{\nnn}{\nn\\}
\newcommand{\FB}{\tx{FB}}
\newcommand{\TX}{\tx{TX}}
\newcommand{\RX}{\tx{RX}}
\renewcommand{\mod}{\tx{mod}}
\newcommand{\m}[1]{\mathbf{#1}}
\newcommand{\td}[1]{\tilde{#1}}
\newcommand{\sbf}[1]{\scriptsize{\textbf{#1}}}
\newcommand{\stxt}[1]{\scriptsize{\textrm{#1}}}
\newcommand{\suml}[2]{\sum\limits_{#1}^{#2}}
\newcommand{\sumlk}{\sum\limits_{k=0}^{K-1}}
\newcommand{\eqhsp}{\hspace{10 pt}}
\newcommand{\tx}[1]{\texttt{#1}}
\newcommand{\Hz}{\ \tx{Hz}}
\newcommand{\sinc}{\tx{sinc}}
\newcommand{\tr}{\mathrm{tr}}
\newcommand{\diag}{\mathrm{diag}}
\newcommand{\MAI}{\tx{MAI}}
\newcommand{\ISI}{\tx{ISI}}
\newcommand{\IBI}{\tx{IBI}}
\newcommand{\CN}{\tx{CN}}
\newcommand{\CP}{\tx{CP}}
\newcommand{\ZP}{\tx{ZP}}
\newcommand{\ZF}{\tx{ZF}}
\newcommand{\SP}{\tx{SP}}
\newcommand{\MMSE}{\tx{MMSE}}
\newcommand{\MINF}{\tx{MINF}}
\newcommand{\RC}{\tx{MP}}
\newcommand{\MBER}{\tx{MBER}}
\newcommand{\MSNR}{\tx{MSNR}}
\newcommand{\MCAP}{\tx{MCAP}}
\newcommand{\vol}{\tx{vol}}
\newcommand{\ah}{\hat{g}}
\newcommand{\tg}{\tilde{g}}
\newcommand{\teta}{\tilde{\eta}}
\newcommand{\heta}{\hat{\eta}}
\newcommand{\uh}{\m{\hat{s}}}
\newcommand{\eh}{\m{\hat{\eta}}}
\newcommand{\hv}{\m{h}}
\newcommand{\hh}{\m{\hat{h}}}
\newcommand{\Po}{P_{\mathrm{out}}}
\newcommand{\Poh}{\hat{P}_{\mathrm{out}}}
\newcommand{\Ph}{\hat{\gamma}}
\newcommand{\mat}[1]{\begin{matrix}#1\end{matrix}}
\newcommand{\ud}{^{\dagger}}
\newcommand{\C}{\mathcal{C}}
\newcommand{\nn}{\nonumber}
\newcommand{\nInf}{U\rightarrow \infty}

\title{Integrated Sensing, Communication, and Computation for Over-the-Air Federated Edge Learning}

\author{Dingzhu Wen$^{\orcidlink{0000-0003-0538-5811}}$,~\IEEEmembership{Member,~IEEE,} Sijing Xie$^{\orcidlink{0009-0006-2398-4657}}$, Xiaowen Cao$^{\orcidlink{0000-0003-4164-071X}}$,~\IEEEmembership{Member,~IEEE,} Yuanhao Cui$^{\orcidlink{0000-0001-6323-8559}}$,~\IEEEmembership{Member,~IEEE,} Jie Xu$^{\orcidlink{0000-0002-4854-8839}}$,~\IEEEmembership{Fellow,~IEEE,} Yuanming Shi$^{\orcidlink{0000-0002-1418-7465}}$,~\IEEEmembership{Senior~Member,~IEEE,} and Shuguang Cui$^{\orcidlink{0000-0003-2608-775X}}$,~\IEEEmembership{Fellow,~IEEE}

\thanks{The work was supported in part by NSFC with Grant No. 62293482, the Basic Research Project No. HZQB-KCZYZ-2021067 of Hetao Shenzhen-HK S\&T Cooperation Zone, the Shenzhen Outstanding Talents Training Fund 202002, the Guangdong Research Projects No. 2017ZT07X152 and No. 2019CX01X104, the Guangdong Provincial Key Laboratory of Future Networks of Intelligence (Grant No. 2022B1212010001), the Shenzhen Key Laboratory of Big Data and Artificial Intelligence (Grant No. ZDSYS201707251409055), the National Natural Science Foundation of China under grants No. 62401369, 62271318, 62471424, 92267202, Shenzhen Science and Technology Program under Grant No. RCBS20231211090520032, and the Shanghai Sailing Program under Grant 23YF1427400. 
Part of the described research work is conducted in the Core Facility Platform of Computer Science and Communication provided by ShanghaiTech University.
An earlier version of this paper was presented at the 2024 IEEE/CIC International Conference on Communications in China (ICCC) [DOI: 10.1109/ICCC62479.2024.10681942].
 }

\thanks{D. Wen, S. Xie, and Y. Shi are with the School of Information Science and Technology, ShanghaiTech University, Shanghai 201210, China (e-mail: \{wendzh,~xiesj2023,~shiym\}@shanghaitech.edu.cn). 
   }

\thanks{X. Cao is with the College of Electronic and Information Engineering, Shenzhen University and Guangdong Provincial Key Laboratory of Future Networks of Intelligence, Shenzhen  518172, China (email: caoxwen@szu.edu.cn), X. Cao is the corresponding author.}

\thanks{Y. Cui is with the School of Information and Communication Engineering, Beijing University of Posts and Telecommunications, Beijing 100876, China (e-mail: cuiyuanhao@bupt.edu.cn).}

\thanks{J. Xu and S. Cui are with the School of Science and Engineering (SSE), the Shenzhen Future Network of Intelligence Institute (FNii-Shenzhen), and the Guangdong Provincial Key Laboratory of Future Networks of Intelligence, the Chinese University of Hong Kong at Shenzhen, Guangdong 518172, China (e-mail: \{xujie,~shuguangcui\}@cuhk.edu.cn).  }

}
  
\maketitle

\begin{abstract}
This paper studies an over-the-air federated edge learning (Air-FEEL) system with integrated sensing, communication, and computation (ISCC), in which one edge server coordinates multiple edge devices to wirelessly sense the objects and use the sensing data to collaboratively train a machine learning model for recognition tasks. In this system, over-the-air computation (AirComp) is employed to enable one-shot model aggregation from edge devices. Under this setup, we analyze the convergence behavior of the ISCC-enabled Air-FEEL in terms of the loss function degradation, by particularly taking into account the wireless sensing noise during the training data acquisition and the AirComp distortions during the over-the-air model aggregation. The result theoretically shows that sensing, communication, and computation compete for network resources to jointly decide the convergence rate. 
Based on the analysis, 
we design the ISCC parameters under the target of maximizing the loss function degradation while ensuring the latency and energy budgets in each round. The challenge lies on the tightly coupled processes of sensing, communication, and computation among different devices. 
To tackle the challenge, we derive a low-complexity ISCC algorithm by alternately optimizing the batch size control and the network resource allocation. 
It is found that for each device, less sensing power should be consumed if a larger batch of data samples is obtained and vice versa. Besides, with a given batch size, the optimal computation speed of one device is the minimum one that satisfies the latency constraint. Numerical results based on a human motion recognition task verify the theoretical convergence analysis and show that the proposed ISCC algorithm well coordinates the batch size control and resource allocation among sensing, communication, and computation to enhance the learning performance.
\end{abstract}
\vspace{-15pt}
\section{Introduction}

Future wireless communication systems are envisioned to deploy artificial intelligence (AI) at the network edge for providing pervasive intelligent services, giving rise to a new research area called edge AI \cite{9606720}. Edge AI tasks naturally involve three tightly coupled processes, i.e., sensing for data acquisition, communication for data sharing, and computation for signal processing and intelligent decision-making \cite{shi2023task,wen2024integrated,wen2024survey}. Integrated sensing and communication (ISAC) \cite{liu2022integrated} and over-the-air federated edge learning (Air-FEEL) \cite{8870236, amiri2020machine, yang2020federated, 9606731} are two emerging techniques that integrate two of the three processes for improving the operational efficiency of edge AI. The former shares hardware and network resources (e.g., spectrum) between wireless sensing and communication, leading to decreased hardware size and enhanced resource utilization efficiency for data acquisition and sharing. The latter conducts efficient model aggregation in FEEL via utilizing an integrated communication and computation technique called over-the-air computation (AirComp) \cite{8468002} to accelerate the training of edge AI models. To exploit both benefits and further enhance resource utilization efficiency, integrated sensing, communication, and computation (ISCC) unifies ISAC and Air-FEEL to break the barriers between the three processes. \textcolor{black}{\cite{10294279} proposed ISCC-based over-the-air (ISCCO) framework, focusing on omnidirectional and directional beampattern designs to minimize AirComp error.} To this end, ISCC-based Air-FEEL is the theme of this paper.

\subsection{Related Works}

ISAC and Air-FEEL have been investigated independently in the literature. In ISAC, one framework is radar-communication coexistence, where existing sensing and communication systems are allowed to coexist via spectrum sharing. The key lies in interference management via designing methods such as null-space radar precoders (see, e.g., \cite{7814210}), compressive-sensing-based receivers (see, e.g., \cite{zheng2018adaptive}), and joint transceivers \cite{chen2022generalized}. The other framework is dual-functional radar communication (DRFC), which aims to develop integrated systems capable of communication and sensing on one hardware. In DFRC, tremendous research has been conducted ranging from dual-functional waveform design \cite{9695365, 10086626}, signal processing \cite{huang2020majorcom}, and radio resource management \cite{wang2021power} to information-theoretic analysis \cite{10153971, 10217169} for unifying the design of sensing and communication. On the other hand, Air-FEEL emerges as a promising technique for dealing with the large communication overhead caused by local model aggregation in FEEL. 
It allows multiple devices to transmit their local models/gradients over the same resource blocks and utilizes the waveform superposition property for calculating the weighted sum of all transmitted symbols \cite{8468002, wen2019reduced, cao2020optimized, Cao_24WC}. The delay of local model aggregation in Air-FEEL is shown to be irrelevant to the number of participating devices, thus leading to a significant enhancement of communication efficiency \cite{8870236, amiri2020machine,yang2020federated,9606731}. 

Recently, researchers have made efforts to apply ISAC to empower edge AI. To be specific, edge inference tasks with sensing data as inputs are accelerated via the dual use of wireless signals for obtaining and transmitting real-time data in \cite{10521597}. Besides, the authors in \cite{10005142} proposed a novel edge training scheme that edge devices transmit abundant on-device sensing data generated via ISAC to a central server for training AI models. Moreover, some researchers have implemented ISAC in FEEL systems. The authors in \cite{9659826} proposed a FEEL approach called Wifederated to train machine learning models efficiently for WiFi sensing tasks while preserving data privacy. To boost the spectrum and hardware utilization efficiency, the ISAC signal was adopted for both sensing and communication in \cite{9792281}, and a cooperative sensing scheme based on vertical FEEL was proposed to enhance the sensing performance. To relieve communication overhead, a feature fusion method was designed in the personalized FEEL framework for ISAC-based gesture recognition systems with real Wi-Fi data \cite {10283763}. Furthermore, the authors in \cite{10506079} considered the joint optimization of multi-task learning resource allocation, including beamforming in ISAC, combined with the FEEL framework.

However, prior works on Air-FEEL and ISAC-supported FEEL fall short in integrating only two of the three processes for improving the efficiency of FEEL. 
As pointed out by \cite{shi2023task, wen2024integrated}, the processes of sensing for data acquisition, computation for local updating, and communication for model sharing are highly coupled in FEEL tasks. On one hand, they compete for network resources like time, bandwidth, and energy. On the other hand, the learning performance depends on the obtained and processed number of samples and the gradient distortion caused during the three processes. To this end, an ISCC-based FEEL framework was proposed in \cite{liu2022toward}. \textcolor{black}{Authors in \cite{10283720} theoretically analyzed the convergence performance of the proposed  FedAVG-ISCC framework. }However, the influence of sensing distortion on the learning performance is ignored in \cite{liu2022toward, 10283720}, which degrades its performance, especially in the scenario when the sensing signal-to-noise ratio (SNR) becomes low. \textcolor{black}{Besides, traditional orthogonal multiple access (OMA) techniques were adopted in \cite{liu2022toward} for transmission, leading to a linearly increasing communication latency as the number of devices increases.} To tackle the above two issues, a practical ISCC-based Air-FEEL scheme is proposed in this work by considering the influence of sensing noise and adopting AirComp technique for alleviating communication overhead.

\subsection{Contributions}

In this paper, an ISCC-based Air-FEEL framework is considered. In each training round, a global model is first shared by the server to all devices, which then acquire real-time data via wireless sensing. Next, all local models are updated on devices using the obtained local data samples. Finally, the server aggregates all updated local models using AirComp for updating the global model. To enhance the learning performance, we target reducing the convergence latency. However, there are two design challenges. One is the tight coupling of sensing, computation, and communication processes in three aspects: 1) the convergence rate depends on the gradient distortion of each round,  which is decided by the sensing and AirComp SNRs, and the number of sensing data samples (i.e., batch size); 2)  the batch size decides the on-device computation load; 3) the sensing, computation, and AirComp (communication) processes compete for network resources like energy and time. 
{\color{black}Next, although the influence of sensory data noise on the convergence performance has been studied in \cite{Mashhadi2014Collaborative}, the mathematical characterization regarding how wireless sensing quality impacts the data noise and the influence of the interplay among the sensing, communication, and computation qualities on the 
convergence performance is still missing, resulting in a lack of theoretical guidance for minimizing the convergence latency.} To tackle these challenges, an ISCC resource allocation scheme is proposed.
The detailed contributions are listed below. 

\begin{itemize}

\item {\bf ISCC-based Air-FEEL Framework}:  A practical ISCC-based Air-FEEL framework is established that elaborates on the processes of sensing for data acquisition, on-device computation for local updating, and AirComp for global aggregation in an online FEEL task. Particularly, the sensing echo signal processing to generate training data samples is introduced. The influence of sensing noise on each data sample is mathematically modeled. The coupling mechanism among sensing, computation, and communication is characterized.

\item {\bf Theoretical Convergence Analysis}: First, we mathematically analyze the influence of sensing and AirComp SNRs as well as the batch size of the obtained sensing data on local gradient distortion, by utilizing the first-order Taylor approximation of the training loss function. 
Accordingly, convergence analysis is conducted.  To the best of our knowledge, this is the first work to characterize the influence of data sample corruption caused by sensing noise on the convergence performance. It is proved that the convergence is accelerated with higher sensing and AirComp SNRs as well as a larger total batch size of all devices. Particularly, a larger sum batch size and a higher sensing SNR promote the per-round loss function degradation in a multiplicative manner. 
This provides a foundation for optimizing the learning performance via network resource allocation among sensing, communication, and computation.

\item {\bf ISCC Scheme of Joint Batch Size Control and Resource Allocation}: To achieve fast convergence with limited network resources, we develop a systematic ISCC design.
We aim at maximizing the loss function degradation in each round subject to the constraints on latency constraints and individual energy budgets at devices. The problem is non-convex and complicated due to the tight coupling of sensing, communication, and computation among different devices. A low-complexity ISCC scheme is proposed to address this problem by alternately solving two convex sub-problems for batch size control and resource allocation, respectively. Benefiting from the derived closed-form solutions, the computational complexity of each sub-problem is $\mathcal{O}(K^2)$ with $K$ being the number of devices. Moreover, it is found that acquiring a larger batch size of data samples on one device leads to smaller sensing power and vice versa due to limited latency and energy budgets. With given batch sizes, the optimal computation speed of one device is the minimum one that satisfies the latency constraint.

\item {\bf Performance Evaluation}: Extensive simulations based on a human motion recognition task \cite{liu2022toward} are conducted to evaluate the performance of the proposed scheme. It is observed that a faster convergence rate is achieved with higher sensing and AirComp SNRs and larger batch sizes, thus verifying the theoretical analysis. Then, the superiority of the proposed scheme is demonstrated over other baselines that optimize the system parameters from a partial view. It is shown that the proposed scheme achieves higher testing accuracy than the baseline schemes with the same delay and energy budgets. This is because the proposed scheme properly exploits the inherent relationship among sensing, communication, and computation via joint batch size control and network resource management. 
\end{itemize}

\vspace{-15pt}
\section{System Model}

\subsection{ISCC-based Air-FEEL System}

\begin{figure*}[h]
  \centering
  \includegraphics[width=0.8\textwidth]{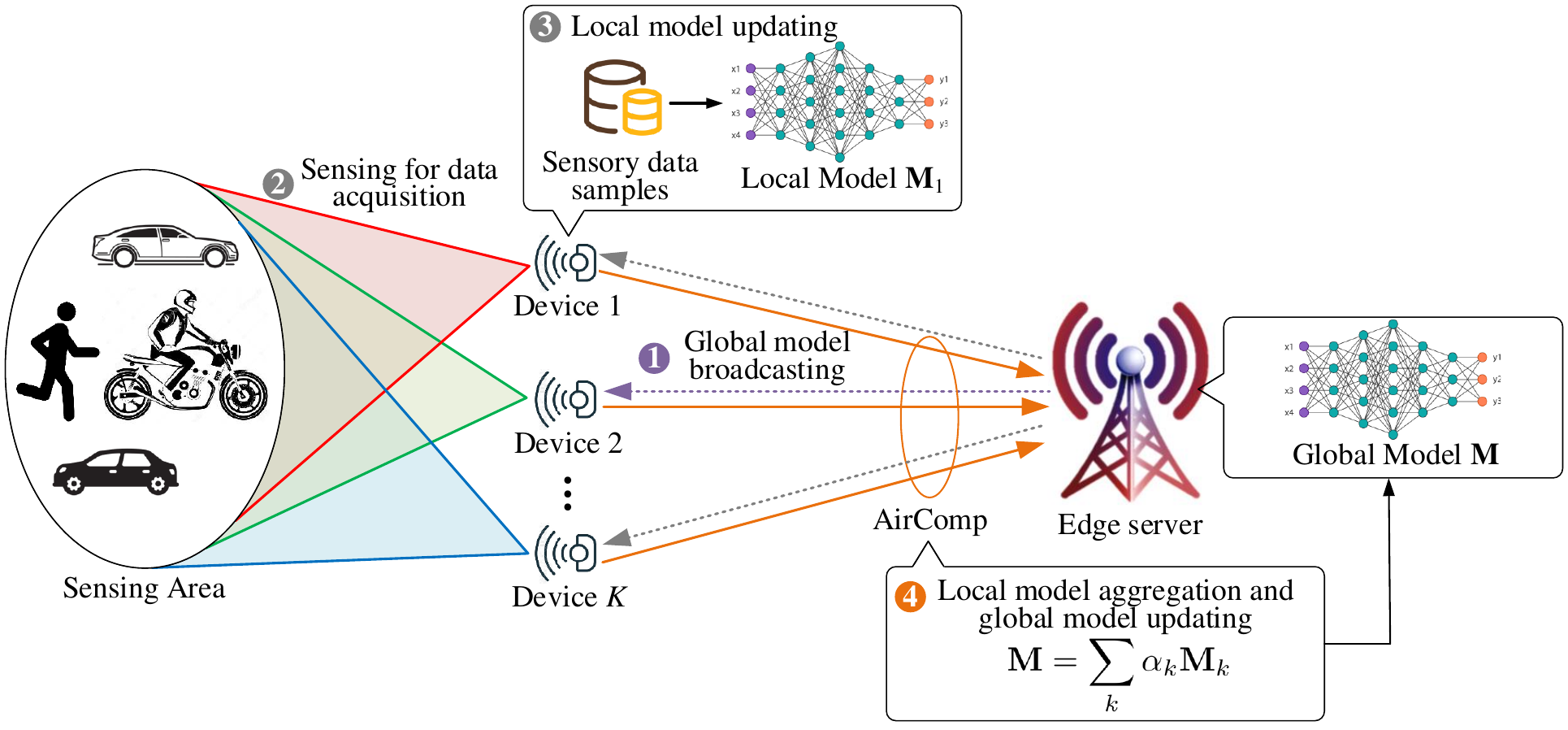}
  \caption{An ISCC-based Air-FEEL system.}
  \label{Fig:SystemModel}
\end{figure*}

As shown in Fig. \ref{Fig:SystemModel}, in the ISCC-based Air-FEEL system, there is one edge server equipped with a single-antenna access point and $K$ edge devices. All devices are each equipped with a single-antenna ISAC transceiver, and thus can shift between the wireless sensing mode and the communication mode in a time-division manner using a shared radio-frequency front-end circuit \cite{han2013joint}. In the sensing mode, one device transmits a dedicated waveform called frequency modulation continuous wave (FMCW) and receives its echo signal as the sensing data, which contains the desired information for training AI models. Different devices are assumed to sense the same target to get homogeneous data samples \cite{wen2023task,zhuang2023integrated}. Without loss of generality, it is assumed that there is no interference among the sensing signals of different devices by assigning orthogonal bandwidths for different devices or placing them in separated locations \cite{wen2024task,liu2022toward}.  In the communication mode, linear analog modulation is applied. At each subcarrier, all devices transmit one symbol at the same time. The technique of AirComp is adopted for aggregating the signals from all devices at the server \cite{8468002,wen2019reduced,cao2020optimized}. Specifically, a functional value (e.g., the weighted sum) is directly calculated at the server by utilizing the waveform superposition property instead of decoding each data stream.

In the Air-FEEL system shown in Fig. \ref{Fig:SystemModel}, the server coordinates edge devices to cooperatively train a global machine learning model. The training loss function is given by
\begin{equation}
    {F}({\bf w}) = \dfrac{1}{K} \sum\limits_{k=1}^K \mathbb{E}_{{\bf x}_k \sim \mathcal{P}_k}\left[ \mathcal{F}({\bf w}; {\bf x}_k) \right],
\end{equation}
where $ \mathcal{F}({\bf w}; {\bf x}_k)$ is the local loss function of device $k$, $\mathcal{P}_k$ is the corresponding local data distribution, ${\bf x}_k \sim \mathcal{P}_k$ represents a random seed with distribution $\mathcal{P}_k$ with each realization representing one data sample. There are four steps to complete an arbitrary training round $t$, as described below.
 
\begin{itemize}
\item \emph{Global Model Sharing}: The server broadcasts a global model ${\bf w}^{(t)}$ to all devices.

\item \emph{Training Samples Acquisition}: All devices conduct wireless sensing for obtaining a batch of sensing data samples $\{\hat{\bf x}_{k,i}^{(t)}\}$ used for local training in this round. The batch size is denoted as $b_k^{(t)}$, which can be adaptively adjusted over different rounds.

\item \emph{Local Model Updating}: Each device calculates the local gradient vector based on the global model and the obtained local training samples $\{\hat{\bf x}_{k,i}^{(t)}\}$. The local gradient vector of device $k$ is denoted as $\tilde{\bf g}_k^{(t)}$.

\item \emph{Aggregation and Global Updating}: The server aggregates all local gradient vectors using the technique of AirComp to obtain a global one, denoted as $\tilde{\bf g}^t$. The global model is then updated as
    ${\bf w}^{(t+1)} = {\bf w}^{(t)} - \alpha^{(t)} \tilde{\bf g}^{(t)}$, 
where $\alpha^{(t)}$ is the learning rate.

\end{itemize}
The above steps iterate until model convergence. 

\subsection{Sensing Model for Training Samples Acquisition}

In each round $t$, each device periodically senses\footnote{To address unlabeled sensing samples, \cite{9659826} proposed methods for continuous annotation, enabling clients to capture accurate labels over time. When sensor data is unavailable, Channel State Information (CSI) is utilized for user tracking, ensuring ongoing training of the federated model. \textcolor{black}{Further, to avoid the high cost and complexity associated with manual annotation, self-supervised learning techniques have emerged to handle unlabeled sensing data, which learn meaningful representations and capture the temporal structure inherent from unlabeled data \cite{10684811}}.} the target by transmitting an FMCW with multiple up-chirps and receiving the echo signal for obtaining the training data samples \cite{liu2022toward}. The echo signal of each period corresponds to one training data sample. The period is denoted as $\tau_s$. The sensing power of each device is assumed to be static in a training round. The echo signal is decomposed of three components, i.e., the desired line-of-sight echo signal containing the useful micro-Doppler information, the clutter caused by higher-order (indirect)
reflected paths, and the sensing noise \cite{wen2024task,liu2022toward,zhuang2023integrated}. The following steps are performed on the echo signal of each period to create a training data sample \cite{wen2024task,liu2022toward}. 
\begin{enumerate}
    \item \emph{Sampling and Reshaping}: The echo signal is sampled into a discrete sequence, which is reformed as a complex Radar matrix in a time-frequency representation \cite{chen2002time}.

    \item \emph{Clutter Cancellation}: A singular value decomposition (SVD) based filter is applied to the Radar matrix for clutter cancellation by excluding the singular values of designated dimensions.

    \item \emph{Short-time Fourier Transform (STFT)}: The matrix after cancellation is processed with STFT to create the  2-dimensional Doppler-time spectrogram.


    \item \emph{Vectorization and Normalization}: The spectrogram is further reshaped into a real vector and normalized by the sensing power.
\end{enumerate}
Since the SVD-based filtering and STFT are linear, the $i$-th data sample obtained by device $k$ in round $t$ is 
\begin{equation}\label{Eq:SensoryDataSample}
\hat{\bf x}_{k,i}^{(t)} = {\bf x}_{k,i}^{(t)} + {\bf c}_k^{(t)} + \dfrac{{\bf n}_s^{(t)}  }{\sqrt{P_{k,s}^{(t)}} }, \; 1 \leq i \leq b_k^{(t)},\; 1\leq k \leq K,
\end{equation}
where ${\bf x}_{k,i}^{(t)}$ is the ground-true data sample, ${\bf c}_k^{(t)}$ is the residual clutter after cancellation, ${\bf n}_s^{(t)}$ is the 
sensing noise with a zero-mean Gaussian distribution and  $\mathbb{E}({{\bf n}_s^{(t)}}^H{\bf n}_s^{(t)}) =\delta_s^2$, $P_{k,s}^{(t)}$ is the sensing power of device $k$, $b_k^{(t)}$ is the batch size of the obtained sensing dataset of device $k$. Without loss of generality, ${\bf c}_k^{(t)}$ follows a zero-mean multi-variate Gaussian distribution \cite{wen2024task,liu2022toward} with 
$\mathbb{E}({{\bf c}_k^{(t)}}^H{\bf c}_k^{(t)}) = \delta_{k,c}^2$. The time to obtain one data sample is denoted as $\tau_s$. The sensing time of device $k$ is given by 
\begin{equation}\label{Eq:SensingTime}
T_{k,s}^{(t)} = b_k^{(t)} \tau_s, \; 1\leq k \leq K.
\end{equation}
The sensing energy consumption of device $k$ is given by 
\begin{equation}\label{Eq:SensingEnergy}
E_{k,s}^{(t)} = P_{k,s}^{(t)} b_k^{(t)} \tau_s, \; 1\leq k \leq K.
\end{equation}
From \eqref{Eq:SensoryDataSample} to \eqref{Eq:SensingEnergy}, one can observe that more time and energy are consumed to obtain a larger batch of data samples. Also, acquiring higher-quality data samples with high sensing SNR consumes more energy.


\subsection{Local Computation Model}
In each training round $t$, the method of gradient descent is conducted on all devices for getting local gradient vectors. Consider an arbitrary device $k$, for which the stochastic gradient vector in terms of a data sample $\hat{\bf x}_{k,i}^{(t)}$ is
\begin{equation}\label{Eq:StochasticGradient}
    \hat{\bf g}_{k,i}^{(t)} = \nabla_{{\bf w}_k^{(t)}} \mathcal{F}({\bf w}_k^{(t)}; \hat{\bf x}_{k,i}^{(t)}), \; 1 \leq i \leq b_k^{(t)},\; 1\leq k \leq K,
\end{equation}
where $\hat{\bf g}_{k,i}^{(t)} \in \mathbb{R}^N$ with $N$ being the length of gradient vectors.
The local gradient vector of device $k$ is then given by
\begin{equation}\label{Eq:LocalGradientVecor}
\tilde{\bf g}_k^{(t)} = \dfrac{1}{b_k^{(t)}} \sum\limits_{i=1}^{b_k^{(t)} } \hat{\bf g}_{k,i}^{(t)},
\end{equation}
That says,
\begin{equation}\label{Eq:LocalGradientElement}
    \tilde{g}_{k,j}^{(t)} =  \dfrac{1}{b_k^{(t)}} \sum\limits_{i=1}^{b_k^{(t)} }\hat{g}_{k,i,j}^{(t)}, \; 1\leq j \leq N,
\end{equation}
where $\tilde{g}_{k,j}^{(t)} $ and $\hat{g}_{k,i,j}^{(t)}$ are $j$-th dimension of $\tilde{\bf g}_k^{(t)}$ and $\hat{\bf g}_{k,i}^{(t)}$, respectively. The computation time for calculating the local gradient vector is
\begin{equation}\label{Eq:ComputationTime}
T_{k, c}^{(t)} = \dfrac{b_k^{(t)} C}{f_k}, \; 1\leq k \leq K,
\end{equation}
where  $C$ is the number of processor operation cycles to update the local model using one data sample and $f_k$ is the computation frequency in cycle/s. The local computation energy consumption by device $k$ is given by
\begin{equation}
E_{k,c}^{(t)} = \Omega_k T_{k, c}^{(t)} f_k^3,\; 1\leq k \leq K,
\end{equation}
where  $\Omega_k$ is a constant characterizing the local computation performance of the processor on device $k$ \cite{zeng2021energy}. By substituting \eqref{Eq:ComputationTime}, the local computation energy consumption is
\begin{equation}\label{Eq:ComputationEnergy}
E_{k,c}^{(t)} = \Omega_k b_k^{(t)} C f_k^2,\; 1\leq k \leq K.
\end{equation}

\subsection{Aggregation via AirComp}
Broadband AirComp is adopted for aggregating all local gradient vectors. 
$M$ orthogonal subcarriers are used for transmission. Over each subcarrier, all local gradient elements of the same dimension are simultaneously transmitted and aggregated at the server. Specifically, for the $j$-th dimension of the gradient vector, it is aggregated as 
\begin{equation}
 y_j = \sum\limits_{k=1}^K h_k^{(t)} \sqrt{\beta_{k,u}^{(t)}} \dfrac{b_k^{(t)}}{b^{(t)}} \tilde{g}_{k,j}^{(t)} + n_u^{(t)}, \; 1\leq j \leq N,
\end{equation}
where $b^{(t)} = \sum\nolimits_{k=1}^K b_k^{(t)}$ is the total data samples used for training in this round, $\beta_{k,u}^{(t)}$ denotes the transmit power for uploading one gradient element, 
$h_k^{(t)}$ is the channel magnitude, $\tilde{g}_{k,j}^{(t)}$ is the $j$-th local gradient element of device $k$, and $n_u^{(t)}$ is the uplink channel noise. To facilitate AirComp, a common approach of channel alignment is performed (see, e.g., \cite{8468002,wen2019reduced,cao2020optimized}), i.e., 
\begin{equation}\label{Eq:ChannelAlignment}
\sqrt{\beta_{k,u}^{(t)}} h_k^{(t)} = \sqrt{\eta^{(t)}}, \; 1\leq k\leq K,
\end{equation}
where $\sqrt{\eta^{(t)}}$ is a positive constant representing the receive signal magnitude. The gradient is thus recovered as
\begin{equation}\label{equ:recov_gradient}
\tilde{g}_j^{(t)} = \dfrac{1}{\sqrt{ \eta^{(t)} } } y_j = \sum\limits_{k=1}^K \dfrac{b_k^{(t)}}{b^{(t)}} \tilde{g}_{k,j}^{(t)} + \dfrac{n_u^{(t)}}{ \sqrt{\eta^{(t)}} }.
\end{equation}
\textcolor{black}{Note that devices with weak channel gains, i.e., $h_k^{(t)}$ is small, should not participate in federated learning to avoid high receive distortion caused by channel fading and noise, which can be observed from the second term in \eqref{equ:recov_gradient}. Furthermore, this results in a lower AirComp SNR $\dfrac{\eta^{(t)}}{\delta_u^2}$, causing a gradual decrease in the loss function and a reduced convergence rate, which will be shown in \eqref{Eq:Per-RoundLossReduction} and \eqref{theorem:con}, respectively.} In this work, we assume that device scheduling has been adopted to exclude all the devices with weak channels in advance \footnote{\textcolor{black}{For the device scheduling, it is necessary to consider the device's overall performance, including batch size, SNR, and computation ability, which will be investigated in future studies.}}. All involved devices have good channel conditions. The recovered gradient vector is 
\begin{equation}\label{Eq:GlobalGradientVector}
\tilde{\bf g}^{(t)} = \sum\limits_{k=1}^K \dfrac{b_k^{(t)}}{b^{(t)}} \tilde{\bf g}_{k}^{(t)} + \dfrac{{\bf n}_u^{(t)}}{ \sqrt{ \eta^{(t)} } },
\end{equation}
where ${\bf n}_u^{(t)}$ is the additive white Gaussian noise vector with  $\mathbb{E}({{\bf n}_u^{(t)}}^H{\bf n}_u^{(t)}) = \delta_u^2$.

The uploading time is then given as
\begin{equation}\label{Eq:UploadingTime}
T_{k,u}^{(t)} = T_{u}^{(t)} = \left\lceil  \dfrac{N}{M} \right\rceil \tau_u, \; 1\leq k \leq K,
\end{equation}
where $\lceil \cdot \rceil $ represents rounding up, $\tau_u$ is the time of transmitting one gradient element in one subcarrier. The channels over the $M$ subcarriers in the time duration $T_{k,u}^{(t)}$ are assumed to be static, since $\tau_u$ and the bandwidth of subcarrier are far smaller than the coherence time and bandwidth \cite{8870236}. The energy consumption of device $k$ in this step is 
\begin{equation}
E_{k,u}^{(t)} =  \mathbb{E}\left[ \left( \sqrt{ \beta_{k,u}^{(t)} } \dfrac{b_k^{(t)}}{b^{(t)}} \tilde{g}_{k,j}^{(t)} \right)^2 \right] N \tau_u, \; 1\leq k \leq K.
\end{equation}
Without loss of generality, each local gradient element $\tilde{g}_{k,j}^{(t)} $ is normalized to have a zero mean and unit variance \cite{8870236, 9606731}. 
Thereby, by substituting the channel alignment condition in \eqref{Eq:ChannelAlignment} and the unit symbol variance mentioned above, the uploading energy consumption of device $k$ is derived as

\begin{equation}\label{Eq:UploadingEnergy}
E_{k,u}^{(t)} =   \dfrac{ \eta^{(t)} {b_k^{(t)}}^2 N \tau_u  }{ H_k^{(t)} {b^{(t)}}^2 }, \; 1\leq k \leq K,
\end{equation}
where $H_k^{(t)} = \mathbb{E}[(h_k^{(t)})^2] $ is the channel power gain.
\vspace{-15pt}
\section{Convergence Analysis of ISCC-based Air-FEEL}\label{Sect:Convergence}

In this section, we provide a convergence analysis for ISCC-enabled Air-FEEL in the presence of both sensing and AirComp errors.

\subsection{Basic Assumptions and Gradient Approximation}

\subsubsection{Basic Assumptions}
First, define the second-order mixed derivatives matrix of the loss function $\mathcal{F}({\bf w}; {\bf x})$ as
\begin{equation}
{\bf H}({\bf w}; {\bf x}) = \nabla_{\bf w} \left( \nabla_{\bf x} \mathcal{F}({\bf w}; {\bf x}) \right), \; \forall {\bf w},{\bf x}.
\end{equation}
which is a part of $\mathcal{F}({\bf w}; {\bf x})$'s Hessian matrix in terms of ${\bf w}$  and ${\bf x}$.
Following \cite{castiglia2022compressed}, the following assumption of bounded norm is made.
\begin{assumption}[Bounded Hessian]\label{Assump:BoundedHessian}
The Frobenius norm of ${\bf H}({\bf w}; {\bf x})$ is bounded:
\begin{equation}
\left\|{\bf H}({\bf w}; {\bf x})\right\|_{\mathcal{F}} \leq A,\; \forall {\bf w}, \forall {\bf x}.
\end{equation}
\end{assumption}

Then, denote the stochastic gradient vector generated by a clean data sample ${\bf x}_{k,i}^{(t)}$ in round $t$ as 
\begin{equation}
{\bf g}_{k,i}^{(t)} = \nabla_{{\bf w}_k^{(t)}}  \mathcal{F}({\bf w}_k^{(t)}; {\bf x}_{k,i}^{(t)}),
\end{equation}
where ${\bf w}_k^{(t)}$ is the local model of device $k$ in round $t$.
Without loss of generality, the following common assumptions are made \cite{allen2018natasha,ghadimi2013stochastic}. 
\begin{assumption}[Unbiased and Variance-Bounded Estimation]\label{Assump:UnbiasedEstimation}
\begin{equation}
\begin{aligned}
&\mathbb{E}({\bf g}_{k,i}^{(t)}) = {\bf g}^{(t)}, \\
&\mathbb{E}\left[\left\| {\bf g}_{k,i}^{(t)} - {\bf g}^{(t)} \right\|^2\right] \leq \sigma^2,
\end{aligned}
\end{equation} 
where ${\bf g}^{(t)}$ is the ground-true gradient vector in round $t$.
\end{assumption}

\begin{assumption}[Bounded Loss Function]
The loss function $\mathcal{F}({\bf w})$ is lower bounded,
	i.e., 
    \begin{equation}\label{eqn:bounded}
		\mathcal{F}({\bf w}) \geq \mathcal{F}^*, \quad \forall {\bf w}.
	\end{equation}
\end{assumption}

\begin{assumption}[Smoothness]\label{Assump:Smooth} 
	The loss function $F(\boldsymbol{z})$ is $ L $-smooth,
	i.e., for any $ {\bf v}, {\bf w} \in \operatorname{dom}(f)$,
	\begin{equation}\label{eqn:smooth}
		\mathcal{F}({\bf v}) \leq \mathcal{F}({\bf w})+({\bf v}- {\bf w})^{\top} \nabla \mathcal{F}({\bf w})+\frac{L}{2}\|{\bf v}-{\bf w}\|^{2}.
	\end{equation} 
\end{assumption}

\subsubsection{First-Order Taylor Approximation of  Gradient}
Recall the stochastic gradient vector $\hat{\bf g}_{k,i}^{(t)}$ defined in \eqref{Eq:StochasticGradient} is generated with a noisy data sample $ \hat{\bf x}_{k,i}^{(t)}$ defined in \eqref{Eq:SensoryDataSample}. 
Then, the Taylor expansion of the local loss function with noisy data at the reference point ${\bf x}_{k,i}^{(t)}$ is given by
\begin{equation}
\begin{aligned}
\mathcal{F}({\bf w}_k^{(t)}; \hat{\bf x}_{k,i}^{(t)}) &= \mathcal{F}({\bf w}_k^{(t)}; {\bf x}_{k,i}^{(t)}) \\
&+ \nabla_{ {\bf x}_{k,i}^{(t)} } \mathcal{F}({\bf w}_k^{(t)}; {\bf x}_{k,i}^{(t)})(\hat{\bf x}_{k,i}^{(t)} - {\bf x}_{k,i}^{(t)}) \\
&+ O\left(\hat{\bf x}_{k,i}^{(t)} - {\bf x}_{k,i}^{(t)}\right),
\end{aligned}
\end{equation}
where $O\left(\hat{\bf x}_{k,i}^{(t)} - {\bf x}_{k,i}^{(t)}\right)$ is the infinitesimal of higher order. By taking the partial derivative with respect to (w.r.t.) ${\bf w}_k^{(t)}$ for both sides of the above equation, it follows that 
\begin{equation}\label{Eq:GradientTaylor}
\hat{\bf g}_{k,i}^{(t)} = {\bf g}_{k,i}^{(t)} + \hat{\bf H}_{k,i}^{(t)} (\hat{\bf x}_{k,i}^{(t)} - {\bf x}_{k,i}^{(t)}) + O\left(\hat{\bf x}_{k,i}^{(t)} - {\bf x}_{k,i}^{(t)}\right),
\end{equation}
where 
\begin{equation}
    \hat{\bf H}_{k,i}^{(t)} = \nabla_{ {\bf w}_k^{(t)} } \left( \nabla_{ {\bf x}_{k,i}^{(t)} } \mathcal{F}({\bf w}_k^{(t)}; {\bf x}_{k,i}^{(t)}) \right).
\end{equation}
According to Assumption \ref{Assump:BoundedHessian}, 
\begin{equation}\label{Eq:BoundedH}
   \left\| \hat{\bf H}_{k,i}^{(t)} \right\|_{\mathcal{F}} = \left\| {\bf H}({\bf w}_k^{(t)}; {\bf x}_{k,i}^{(t)}) \right\|_{\mathcal{F}} \leq A.
\end{equation}
Next, by ignoring the infinitesimal of higher order terms in \eqref{Eq:GradientTaylor} \cite{castiglia2022compressed}, we have
\begin{equation}
\hat{\bf g}_{k,i}^{(t)} \approx {\bf g}_{k,i}^{(t)} + \hat{\bf H}_{k,i}^{(t)} (\hat{\bf x}_{k,i}^{(t)} - {\bf x}_{k,i}^{(t)}),
\end{equation}
which, by substituting the sensing data sample in \eqref{Eq:SensoryDataSample}, is derived as
\begin{equation}\label{Eq:TaylorApproximation}
\hat{\bf g}_{k,i}^{(t)} \approx {\bf g}_{k,i}^{(t)} +  \hat{\bf H}_{k,i}^{(t)}\left({\bf c}_k^{(t)} + \dfrac{  {\bf n}_s^{(t)}}{\sqrt{P_{k,s}^{(t)}} } \right).
\end{equation}

\subsection{Unbiased and Variance-Bounded Global Gradient Vector}
By substituting \eqref{Eq:LocalGradientVecor} and \eqref{Eq:TaylorApproximation} into the globally recovered gradient vector at the server in \eqref{Eq:GlobalGradientVector}, it is derived as
\begin{equation}\label{Eq:ApproximatedGlobalGradientVector}
\tilde{\bf g}^{(t)} = \dfrac{1}{b^{(t)}} \sum\limits_{k=1}^K \sum\limits_{i=1}^{b_k^{(t)} } \left[  {\bf g}_{k,i}^{(t)} +  \hat{\bf H}_{k,i}^{(t)}\left({\bf c}_k^{(t)} + \dfrac{  {\bf n}_s^{(t)}}{\sqrt{P_{k,s}^{(t)}} } \right)  \right]  + \dfrac{{\bf n}_u^{(t)}}{ \sqrt{ \eta^{(t)} } }.
\end{equation}
Then we have the following lemma on the unbiased and variance-bounded gradient estimation.
\begin{lemma}[Unbiased and Variance-Bounded Global Gradient Vector]\label{Lma:GlobalGradientVector}
\begin{equation}\label{Eq:LmaGlobal}
\begin{aligned}
&\text{Unbiased Estimation}:\; \mathbb{E}\left( \tilde{\bf g}^{(t)}\right) = {\bf g}^{(t)},\\
&\text{Bounded Variance}: \; \\
&\mathbb{E}\left[\!\left\| \tilde{\bf g}^{(t)} \!- \!{\bf g}^{(t)} \right\|_2^2\!\right] \!
\leq \! \sum\limits_{k=1}^K \!\dfrac{b_k^{(t)}}{{b^{(t)}}^2} \!\left[ \!\sigma^2+ A^2(\delta_{k,c}^2+ \!\dfrac{ \delta_s^2}{P_{k,s}^{(t)}})\!\right] \!+\! \dfrac{\delta_u^2}{\eta^{(t)}}.
\end{aligned}
\end{equation}
\end{lemma}
\begin{proof}
In \eqref{Eq:ApproximatedGlobalGradientVector}, the sensing clutter ${\bf c}_k^{(t)}$, the sensing noise ${\bf n}_s^{(t)}$, and the channel noise ${\bf n}_u^{(t)}$ have zero means. According to Assumption \ref{Assump:UnbiasedEstimation}, it is shown that $\mathbb{E}\left( \tilde{\bf g}^{(t)}\right) = {\bf g}^{(t)}$. Based on this result, the variance of $ \tilde{\bf g}^{(t)}$ is given by
\begin{equation}\label{Eq:ServerSideVarianceDerivation}
\begin{aligned}
& \mathbb{E}\left[\left\| \tilde{\bf g}^{(t)} - {\bf g}^{(t)} \right\|_2^2\right] \\
& =  \mathbb{E}\left[\left\|   \dfrac{1}{b^{(t)}} \sum\limits_{k=1}^K \sum\limits_{i=1}^{b_k^{(t)} }{\bf g}_{k,i}^{(t)}   - {\bf g}^{(t)} \right\|_2^2\right] + \mathbb{E}\left[\left\| \dfrac{{\bf n}_u^{(t)}}{ \sqrt{ \eta^{(t)} } }  \right\|_2^2\right] \\
& ~~~+ \mathbb{E}\left[\left\| \dfrac{1}{b^{(t)}} \sum\limits_{k=1}^K \sum\limits_{i=1}^{b_k^{(t)} } \hat{\bf H}_{k,i}^{(t)}\left({\bf c}_k^{(t)} + \dfrac{  {\bf n}_s^{(t)}}{\sqrt{P_{k,s}^{(t)}} } \right) \right\|_2^2\right].\\
\end{aligned} 
\end{equation}
The above equation holds as the expectations of all cross terms are ${\bf 0}$. According to Assumption \ref{Assump:UnbiasedEstimation}, the first term in \eqref{Eq:ServerSideVarianceDerivation} is bounded by
\begin{equation}
 \mathbb{E}\left[\left\|   \dfrac{1}{b^{(t)}} \sum\limits_{k=1}^K \sum\limits_{i=1}^{b_k^{(t)} }{\bf g}_{k,i}^{(t)}   - {\bf g}^{(t)} \right\|_2^2\right] \leq \dfrac{\sigma^2}{b^{(t)}}.
\end{equation}
The second term is $ \mathbb{E}\left[\left\| \dfrac{{\bf n}_u^{(t)}}{ \sqrt{ \eta^{(t)} } }  \right\|_2^2\right] = \dfrac{\delta_u^2}{\eta^{(t)}}$. The third term, according to Assumption \ref{Assump:BoundedHessian}, is bounded by 
\begin{equation}
\begin{aligned}
&\mathbb{E}\left[\left\| \dfrac{1}{b^{(t)}} \sum\limits_{k=1}^K \sum\limits_{i=1}^{b_k^{(t)} } \hat{\bf H}_{k,i}^{(t)}\left({\bf c}_k^{(t)} + \dfrac{  {\bf n}_s^{(t)}}{\sqrt{P_{k,s}^{(t)}} } \right) \right\|_2^2\right]\\
& ~~ \leq  \sum\limits_{k=1}^K \dfrac{b_k^{(t)}}{{b^{(t)}}^2} \left[ A^2(\delta_{k,c}^2+ \dfrac{ \delta_s^2}{P_{k,s}^{(t)}})\right],
\end{aligned}
\end{equation}
since the expectation of the cross terms of the sensing clutter and noise is ${\bf 0}$. In summary, the bounded variance in \eqref{Eq:LmaGlobal} can be derived. The proof is complete.


\end{proof}

\vspace{-20pt}  

\subsection{Convergence Analysis}
In an arbitrary round $t$, the global model is updated as
\begin{equation}\label{Eq:GlobalUpdate}
    {\bf w}^{(t+1)} = {\bf w}^{(t)} - \alpha^{(t)} \tilde{\bf g}^{(t)},
\end{equation}
where $\alpha^{(t)}$ is the learning rate. Based on Assumptions \ref{Assump:BoundedHessian}--\ref{Assump:Smooth} and Lemma \ref{Lma:GlobalGradientVector}, a per-round loss function degradation can be derived as follows.
\begin{lemma} \label{Lem:Per-roundLoss}
Given learning rate $\alpha^{(t)}= \frac{1}{\sqrt{T}L}$, the loss function degradation after an arbitrary training round  $t$ is bounded by 
\begin{align}
\mathbb{E} \left[\mathcal{F}\left( {\bf w}^{(t)}  \right) - \mathcal{F}\left( {\bf w}^{(t+1)}\right) \right] \geq (\frac{1}{ \sqrt{T}L} - \frac{1}{2TL}) \|{\bf g}^{(t)}  \|^2  \notag\\
- \frac{1}{2TL} \bigg\{ \dfrac{\delta_u^2}{\eta^{(t)}} +  \sum\limits_{k=1}^K \dfrac{b_k^{(t)}}{{b^{(t)}}^2} \bigg[ \sigma^2+ A^2(\delta_{k,c}^2+ \dfrac{ \delta_s^2}{P_{k,s}^{(t)}})\bigg] \bigg\},\label{Eq:Per-RoundLossReduction}
\end{align}
where $T$ is the number of global rounds.
\end{lemma}
\begin{proof}
	Please refer to Appendix \ref{_proof_of_lemma_eachround}.
\end{proof}
It is observed in Lemma \ref{Lem:Per-roundLoss} that the loss function decreases more rapidly with higher AirComp SNR $\eta^{(t)}/\delta_u^2$, a larger sum batch size $b^{(t)}$, and higher sensing SNR $\{ P_{k,s}^{(t)}/\delta_s^2 \}$.  Particularly, a larger sum batch size and a higher sensing SNR promote the per-round loss function degradation in a multiplicative manner. 
However, these conditions cannot be simultaneously achieved due to the limited network resources like time, energy, and on-device processing capability, as well as hostile channel fading and noise. Thus, it is desirable to design an adaptive resource allocation scheme among the processes of sensing, on-device computation, and AirComp for minimizing the loss function after each training round.
 
\begin{theorem}\label{theorem:converge}
Based on Lemma \ref{Lem:Per-roundLoss} and given $\alpha^{(t)}= \frac{1}{\sqrt{T}L}$, the ground-truth gradient vector sequence $\left\{ {\bf g}^{(t)}\right\}_{t=0}^T$ converges as follows:
\begin{equation}
    \lim _{T \rightarrow+\infty} \frac{1}{T} \sum_{t=0}^{T-1} \|{\bf g}^{(t)}  \|^2
    \leq G_T,\label{theorem:con}
\end{equation}
where $  G_T 
         = \dfrac{2L}{\sqrt{T}} \left[\mathcal{F}\left( {\bf w}^{(0)} \right)-\mathcal{F}^* \right] + \frac{1}{\sqrt{T} \cdot T} \sum_{t=0}^{T-1}  \bigg\{ \dfrac{\delta_u^2}{\eta^{(t)}} +  \sum\limits_{k=1}^K \dfrac{b_k^{(t)}}{{b^{(t)}}^2} \bigg[ \sigma^2+ A^2(\delta_{k,c}^2+ \dfrac{ \delta_s^2}{P_{k,s}^{(t)}})\bigg] \bigg\}$ and $\lim _{T \rightarrow+\infty} G_T = 0$. 

\end{theorem}
\begin{proof}
	Please refer to Appendix \ref{_proof_of_theorem:converge}.
\end{proof}
\begin{remark}[Task-Oriented ISCC for Air-FEEL]
The sensing, on-device computation, and AirComp processes are highly coupled in the ISCC-enabled Air-FEEL task, as shown in Fig. \ref{Fig:Coupling}. According to Theorem \ref{theorem:converge}, the convergence rate depends on the SNRs of both sensing and AirComp, and the number of obtained and locally updated sensing data samples in each training round. The three processes compete for more time and energy for maximizing their own convergence contribution. Besides, the number of sensed local samples of each device decides its computation cost. Therefore, task-oriented ISCC schemes are desirable for enhancing the convergence performance.
\end{remark}
\begin{figure}[t]
  \centering
  \includegraphics[width=0.5\textwidth]{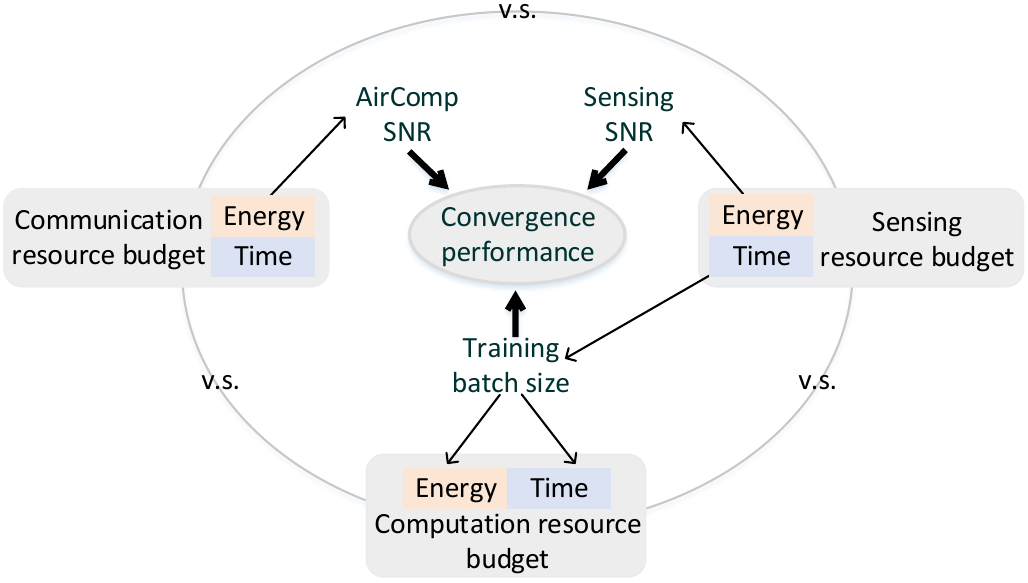}
  \caption{Coupling mechanism of sensing, communication, and computation in Air-FEEL.}
  \label{Fig:Coupling}
\end{figure}

\vspace{-15pt}
\section{Joint Batch Size and Resource Allocation Optimization }

With the obtained convergence analysis of ISCC-enabled Air-FEEL in the preceding section, the joint batch size and resource allocation optimization policies for accelerating the convergence is ready to shown in this section.

\vspace{-15pt}
\subsection{Problem Formulation}
Without loss of generality, we aim at accelerating Air-FEEL's convergence process by maximizing the per-round loss function degradation, i.e., $\mathbb{E} \left[\mathcal{F}\left({\bf w}^{(t)}\right) - \mathcal{F}\left({\bf w}^{(t+1)}\right) \right]$ for each particular round $t$, 
via optimizing the AirComp receive signal power $\eta^{(t)}$, the local and overall batch sizes $\{b_k^{(t)}\}$ and $b^{(t)}$, the sensing power $\{P_{k,s}^{(t)}\}$, and the computation speeds $\{f_k^{(t)}\}$. 

Due to network resource limitations, there are three key constraints on the total batch size, latency, and energy consumption, each of which is introduced below.



\subsubsection{Constraint on Training Batch Size}
The total number of samples obtained by all devices, i.e., the overall batch size, in round $t$ is given by
\begin{equation}
\mathcal{C}_1:\; b^{(t)} = \sum\limits_{k=1}^K b_k^{(t)}.
\end{equation}


\subsubsection{Latency Constraint}
The latency of each device $k$ in round $t$ includes the sensing time defined in \eqref{Eq:SensingTime}, the on-device computation time defined in \eqref{Eq:ComputationTime}, and the uploading time \eqref{Eq:UploadingTime}. The total latency of each device should not exceed the allowed latency for completing this round denoted by $\mathcal{T}$, i.e.,
\begin{equation}
    T_{k,s}^{(t)} + T_{k, c}^{(t)} + T_{u}^{(t)} \leq \mathcal{T}, \; 1 \leq k \leq K,
\end{equation}
which, by substituting \eqref{Eq:SensingTime}, \eqref{Eq:ComputationTime}, and \eqref{Eq:UploadingTime}, is re-expressed as
\begin{equation}
\mathcal{C}_2:\; b_k^{(t)} \tau_s + \dfrac{b_k^{(t)} C}{f_k} + \left\lceil  \dfrac{N}{M} \right\rceil \tau_u \leq \mathcal{T}, \; 1 \leq k \leq K.
\end{equation}

\subsubsection{Energy Consumption Constraint}
The energy consumption of each device $k$ in round $t$ includes the sensing energy consumption in \eqref{Eq:SensingEnergy}, the on-device computation sensing energy consumption in \eqref{Eq:ComputationEnergy}, and the uploading energy consumption in \eqref{Eq:UploadingEnergy}. The total energy consumption of each device should not exceed its budget $E_k$, i.e.,
\begin{equation}
    E_{k,s}^{(t)} +   E_{k,c}^{(t)} + E_{k,u}^{(t)} \leq E_k, \; 1\leq k \leq K,
\end{equation}
which, by substituting \eqref{Eq:SensingEnergy}, \eqref{Eq:ComputationEnergy}, and \eqref{Eq:UploadingEnergy} is given by
\begin{equation}
\mathcal{C}_3:\; P_{k,s}^{(t)} b_k^{(t)} \tau_s + \Omega_k  b_k^{(t)} C f_k^2 + \dfrac{ \eta^{(t)} { b_k^{(t)}}^2 N \tau_u  }{  H_k^{(t)} {b^{(t)}}^2 }  \leq E_k,
\end{equation}
for $1\leq k \leq K$.

Since the exact value of $\mathbb{E} \left[\mathcal{F}\left({\bf w}^{(t)}\right) - \mathcal{F}\left({\bf w}^{(t+1)}\right) \right]$ is not available, a common approach is maximizing its lower
bound approximation defined in the right side of \eqref{Eq:Per-RoundLossReduction} instead (see, e.g., \cite{9606731,zeng2024federated,9252950}). Note that the norm of the ground-true gradient vector $\|{\bf g}^{(t)}  \|$, the number of total training iterations $T$, and the  smoothness parameter $L$ are irrelevant to the ISCC resource allocation, and as a result, maximizing the right side of \eqref{Eq:Per-RoundLossReduction} is equivalent to
\begin{equation*}\label{equ:P_1}
\mathcal{P}_1: 
	\begin{array}{cl}
		\underset
		{
			{
				\substack{ \eta^{(t)},   \{  b_k^{(t)}\}, b^{(t)}, \\ \\  \{P_{k,s}^{(t)} \}, \{f_k^{(t)}\} } 
			}
		}
		{
			\min
		}
		&	
		\dfrac{\delta_u^2}{\eta^{(t)}} +  \sum\limits_{k=1}^K \dfrac{b_k^{(t)}}{{b^{(t)}}^2} \bigg[ \sigma^2+ A^2(\delta_{k,c}^2+ \dfrac{ \delta_s^2}{P_{k,s}^{(t)}})\bigg],
		\\
		\text{s.t.}
		& \mathcal{C}_1 \sim \mathcal{C}_3.
	\end{array}
\end{equation*}
By solving $\mathcal{P}_1$ in all rounds, the convergence is accelerated. In the sequel, we consider an arbitrary round and the
notation $(t)$ is omitted for simplicity.
Note that $\mathcal{P}_1$ is a non-convex problem with highly coupled sensing, computation, and communication parameters among different devices. To tackle it, the following variable transformation is first made:
\begin{equation}
\alpha_k = b_k/b,\; 1\leq k \leq K.
\end{equation}
It follows that $\mathcal{P}_1$ is equivalently derived as
	\begin{align}
		\mathcal{P}_2: 	
 \min\limits_{ \substack{ \eta,   \{  b_k\},  \{ \alpha_k\} \\ \\  \{P_{k,s} \}, \{f_k\} } }
	    ~&	\dfrac{\delta_u^2}{\eta} +  \sum\limits_{k=1}^K \dfrac{\alpha_k^2}{ b_k } \bigg[ \sigma^2+ A^2(\delta_{k,c}^2+ \dfrac{ \delta_s^2 }{ P_{k,s} })\bigg],\notag
		\\
		\text{s.t.}~~~~~
		~& \sum\limits_{k=1}^K \alpha_k =1,\label{equ:P_2alpha}\\
		&  b_k \tau_s + \dfrac{b_k C}{f_k} + \left\lceil  \dfrac{N}{M} \right\rceil \tau_u \leq \mathcal{T}, \; \forall k,\label{equ:P_2time}\\
		& \! P_{k,s} b_k \tau_s +\! \Omega_k  b_k C f_k^2\! +\! \dfrac{\eta \alpha_k^2 N \tau_u }{H_k}  \!   \leq \!E_k,\label{equ:P_2energy}
	\end{align}
for all $k$.
However, $\mathcal{P}_2$ still remains non-convex due to the coupled variables therein. To this end, we decompose it into two convex sub-problems and adopt the method of alternating optimization to find its solution, as described below.

\subsection{Batch Size Optimization} 
The first sub-problem optimizes batch sizes of all devices, i.e. $\{\alpha_k\}$ and $\{b_k\}$, with given variables of $\{P_{k,s}\}$, $\{f_k\}$, and $\eta$, which is reformulated as 
\begin{align}	
\mathcal{P}_3: \min\limits_{  \{  b_k\},  \{ \alpha_k\} }
	    ~&	 \sum\limits_{k=1}^K \dfrac{\alpha_k^2}{ b_k } \bigg[ \sigma^2+ A^2(\delta_{k,c}^2+ \dfrac{ \delta_s^2 }{ P_{k,s} })\bigg],\notag\\
		\text{s.t.}~~~~~& \eqref{equ:P_2alpha},~\eqref{equ:P_2time},~\text{and} ~ \eqref{equ:P_2energy}.\notag
\end{align}
To handle the integer variables, $b_k$ are relaxed to be continuous and are rounded to integers for implementation, with negligible loss due to the typically large magnitudes. It is easy to show that $\mathcal{P}_3$ is convex as all constraints form a convex set and the objective function is convex due to the convexity of $\alpha_k^2/b_k$. 
\subsubsection{Feasibility Checking} 

To begin with, it is necessary to check whether  $\mathcal{P}_3$ is feasible under given network conditions and given variables of $\{P_{k,s}\}$, $\{f_k\}$, and $\eta$. The feasibility can be checked by equivalently solving the following problem: 
\begin{equation}
\begin{array}{cl}
    \underset
    {
        {
            \substack{\{  b_k\}, \{  \alpha_k\}}
        }
    }
    {
        \max 
    }
     &\sum\limits_{k=1}^K \alpha_k,
    \\
    \text{s.t.}     & \eqref{equ:P_2time}~\text{and} ~ \eqref{equ:P_2energy}.
\end{array}
\label{Eq:Feasibility}
\end{equation}
If the optimal value achieved by problem \eqref{Eq:Feasibility} is no smaller than $1$, then problem $\mathcal{P}_3$ is feasible, since all constraints therein can be satisfied. Note that problem \eqref{Eq:Feasibility} is convex and can be solved using existing algorithms such as interior point methods \cite{potra2000interior}.

\subsubsection{Optimal Solution to Problem $\mathcal{P}_3$} 

Notice that problem $\mathcal{P}_3$ is convex and satisfies the Slater's condition, and thus strong duality holds between the primal problem and its dual problem. Therefore, we apply the Karush–Kuhn–Tucker (KKT) condition to optimally solve problem  $\mathcal{P}_3$. 
Then, we have the following lemma, where $\mu$, $\{\gamma_k\}$, $\{\lambda_k\}$ are the dual variables associated with constraints \eqref{equ:P_2alpha}, \eqref{equ:P_2time}, and  \eqref{equ:P_2energy}, respectively.

\begin{lemma}\label{Lem:P3}
The optimal solution to problem $\mathcal{P}_3$ is derived as 
\begin{align}
&\alpha_k^{\star} =  \dfrac{ (- \mu^{\star} - 2\sqrt{  J_k M_k    })H_k }{  2\lambda_k^{\star} \eta  N \tau_u }, \; \forall k,\label{Eq:LemP3:alpha}\\
&b_k^{\star} = \dfrac{H_k}{ 2\lambda_k \eta  N \tau_u } \bigg( -\mu^{\star} \sqrt{ \dfrac{ M_k }{ J_k } }  - 2M_k  \bigg), \; \forall k,\label{Eq:LemP3:b}
\end{align}
with $J_k = \gamma_k^{\star} (\tau_s + C/ f_k ) +\lambda_k^{\star} (   P_{k,s} \tau_s + \Omega_k  C f_k^2 )$ and $M_k =  \sigma^2+ A^2(\delta_{k,c}^2+  \delta_s^2 / P_{k,s} )$, where $\mu^{\star}$, $\{\gamma_k^{\star}\}$, and $\{\lambda_k^{\star}\}$ are the optimal dual variables that can be obtained via the primal-dual algorithm with a complexity of $\mathcal{O}(K^2)$.
\end{lemma}
\begin{proof}
	Please refer to Appendix \ref{Apdx:LemP3}.
\end{proof}
By substituting $\alpha_k^{\star} = b_k^{\star}/b$ into Lemma~\ref{Lem:P3}, the overall batch size is 
\begin{equation}\label{Eq:Insight1}
b^{\star} =  \sqrt{ \dfrac{  \sigma^2+ A^2(\delta_{k,c}^2+  \delta_s^2 / P_{k,s} ) }{ \gamma_k^{\star} (\tau_s + C/ f_k ) +\lambda_k^{\star} (   P_{k,s} \tau_s + \Omega_k  C f_k^2 ) } },\; \forall k.
\end{equation}
\begin{remark}[Optimal Batch Sizes]
According to Lemma~\ref{Lem:P3} and \eqref{Eq:Insight1}, the optimal batch sizes have the following properties. To begin with, the local batch size of device $k$, i.e., $b_k^{\star}$, is proportional to its channel gain $H_k$, as more energy can be used for sensing and computation with a stronger channel. Besides, the overall batch size $b$ decreases with the consumed sensing power to obtain one sample, i.e., $P_{k,s}$, and the computation load to update the local model using one data sample, i.e., $C$. Moreover, both $\{b_k^{\star}\}$ and $b^{\star}$ first increase with the computation speeds $\{f_k\}$ but then decrease, since higher computation speeds lead to lower latency but higher energy consumption.
\end{remark}



\subsection{Resource Allocation Optimization}
The second sub-problem of resource allocation optimizes the variables of  $\{P_{k,s}\}$, $\{f_k\}$, and $\eta$, with given $\{\alpha_k\}$ and $\{b_k\}$, which is given by 
\begin{align}
		\!\!\!\mathcal{P}_4:	\min\limits_{\eta,\{P_{k,s}\}, \{f_k\}}&	
		\dfrac{\delta_u^2}{\eta} +  \sum\limits_{k=1}^K \dfrac{\alpha_k^2}{ b_k } \bigg[ \sigma^2+ A^2(\delta_{k,c}^2+ \dfrac{ \delta_s^2 }{ P_{k,s} })\bigg],
		\\
		\text{s.t.}~~~ ~~& \eqref{equ:P_2time}~\text{and} ~ \eqref{equ:P_2energy}.\notag
\end{align}
In $\mathcal{P}_4$, it is easy to observe that the objective function is the sum of multiple inversely proportional functions and the two constraints form a convex set. Hence, $\mathcal{P}_4$ is convex, which can be optimally solved by using the KKT conditions as shown in the following lemma.
Let $\{\phi_k\}$ and $\{\psi_k\}$ denote the dual variables associated with constraints \eqref{equ:P_2time} and \eqref{equ:P_2energy} in problem $\mathcal{P}_4$, respectively.
\begin{lemma}\label{Lem:P4}
The optimal solution to problem $\mathcal{P}_4$, denoted by $\{P_{k,s}^{\star}\}$, $\{f_k^{\star}\}$, and $\eta^{\star}$, are given by
\begin{align}
&P_{k,s}^{\star} =\dfrac{ \alpha_k A \delta_s}{ b_k \sqrt{ \psi_k^{\star} \tau_s} } = \dfrac{ A \delta_s}{ b \sqrt{ \psi_k^{\star} \tau_s} }, ~\forall k,\label{Lem:P4:P}\\
&\eta^{\star} =\sqrt{ \delta_u^2 \bigg( \sum\limits_{k=1}^K\frac{\psi_k^{\star} \alpha_k^2 N \tau_u  }{  H_k  } \bigg)^{-1} },\label{Lem:P4:eta}\\
&f_k^{\star} = \dfrac{ b_k C}{ {\mathcal{T}-\left\lceil  \dfrac{N}{M} \right\rceil \tau_u-b_k \tau_s } }, \; \forall k,\label{Eq:OptimalFrequency}
\end{align}
where $\{\phi_k^{\star}\}$ and $\{\psi_k^{\star}\}$ are the optimal dual variables that could be obtained by using the primal-dual algorithm with a complexity of $\mathcal{O}(K^2)$. 
\end{lemma}
\begin{proof}
Please refer to Appendix \ref{Apdx:LemP4}.
\end{proof}
\begin{remark}[Optimal Resource Allocation]
According to Lemma \ref{Lem:P4}, the optimal computation speeds of each device $k$, i.e., $f_k^{\star}$, is the minimum one that satisfies the latency constraint, since a lower computation speed also alleviates the computational energy consumption and more energy budgets can be used for sensing and communication for either enhancing their qualities or obtaining more data samples. Besides, for each device $k$, less sensing power $P_{k,s}^{\star}$ should be allocated if a larger batch size $b_k$ is obtained.
\end{remark}


\subsection{Overall Alternating Optimization Algorithm }
Based on the solutions of the two convex sub-problems $\mathcal{P}_3$ and $\mathcal{P}_4$, the alternating optimization is adopted to solve $\mathcal{P}_2$, as summarized in Algorithm \ref{Alg:Solution}. At the beginning of the algorithm, initialization is made to guarantee the feasibility of $\mathcal{P}_3$ via selecting initial values of $\{P_{k,s}\}$, $\{f_k\}$, and $\eta$ that ensure the maximum of the feasibility problem in \eqref{Eq:Feasibility} no less than 1. Then, $\mathcal{P}_3$ and $\mathcal{P}_4$ are sequentially and iteratively solved by fixing the variables of each other. 
Since both $\mathcal{P}_3$ and $\mathcal{P}_4$ are convex and the optimal solutions can be achieved, all subsequent sub-problems are feasible as long as the first one is feasible. {\color{black}The convergence of Algorithm \ref{Alg:Solution} is guaranteed since each step in the iteration leads to a non-increasing objective value in $\mathcal{P}_2$ and the objective functiointegran of $\mathcal{P}_2$ is lower bounded.}

\begin{algorithm}[]
	\caption{Joint Design of Batch Size and Resource Allocation}\label{Alg:Solution}
	\LinesNumbered
	\KwIn{ $H_k$,  $\mathcal{T}$ and \{$E_{k}$\}. } 
	\textbf{Loop}\\ 
     \quad \textbf{Initialize} $\{P_{k,s}\}$, $\{f_k\}$, and $\eta$.\\
	\textbf{Until} $\mathcal{P}_3$ is feasible.\\
    \textbf{Loop}\\   
     \quad Obtain optimal solution to $\mathcal{P}_3$ via Lemma \ref{Lem:P3}.\\
	\quad Obtain optimal solution to $\mathcal{P}_4$ via Lemma \ref{Lem:P4}. 
 
    
    \textbf{Until Convergence}.\\
    \KwOut{ $\{\alpha_k\}$, $\{b_k\}$, \{$P_{k,s}$\}, \{$f_{k}$\}, and $\eta$.}
\end{algorithm}
\vspace{-15pt}
{\color{black}Algorithm \ref{Alg:Solution} needs to be executed to solve $\mathcal{P}_2$, which is equivalent to $\mathcal{P}_1$, only when the network settings change, including channel state information, energy budgets, sensing distortion, and communication noise. Particularly, in the case where the network settings vary over different training rounds, Algorithm \ref{Alg:Solution} should be executed in each round to solve $\mathcal{P}_2$. The complexity of Algorithm 1 depends on alternatively solving the two convex subproblems $\mathcal{P}_3$ and $\mathcal{P}_4$, whose computation complexities are both $\mathcal{O}(K^2)$ with $K$ being the number of devices. Note that the execution of Algorithm \ref{Alg:Solution} is at the server with powerful processors, its computation load can be ignored.}
\vspace{-15pt}
\section{Simulation Results}
\subsection{Simulation Settings}
In the experiment, we consider an Air-FEEL network with one server and $K = 6$ ISAC devices. The downlink channel gains of devices are assumed to be Rayleigh fading. The path loss is $10^{-3}$. We adopt the wireless sensing dataset in \cite{liu2022toward} to train ResNet-10 \cite{he2016deep}, which has 4,900,677 model parameters in total. The sensing task is to classify seven different human motions, i.e., standing, adult pacing, child pacing, adult walking, child walking, adult walking, and child walking. The learning rate is 0.1. Other related parameters are listed in Table \ref{table:para} by default.
\begin{table}[htbp]
	\centering
	\caption{Simulation parameter settings}
	\label{table:para}
	\begin{tabular}{|c|c|c|}
		\hline
		Parameter          & Description                                                              & Value                                  
                   \\ \hline
		$f^{\max}_k$          & \begin{tabular}[c]{@{}c@{}} Maximal CPU frequency \\of device $k$ ($\mathrm{Hz}$)       \end{tabular}                               & $\operatorname{Unif}(0.4,2) \times 10^9 $                                                  \\ \hline
		$p^{\max}_{k,s}$          & \begin{tabular}[c]{@{}c@{}}Maximal sensing \\power of device $ k $ ($\mathrm{W}$)    \end{tabular}                                & $\operatorname{Unif}(0.6,3) \times 10^{-2}$                                                 \\ \hline
		$\Omega_k$            & \begin{tabular}[c]{@{}c@{}}CPU constant\\ of device $ k $   \end{tabular}                                & $\operatorname{Unif}(0.2,1) \times 10^{-26}  $                                         \\ \hline
  		$C$              & \begin{tabular}[c]{@{}c@{}} CPU cycles \\for one sample \end{tabular}                              & $2.5\times 10^{7}$                            
    \\ \hline
		$\sigma^2$            
  & \begin{tabular}[c]{@{}c@{}} Noise spectral \\density  ($\mathrm{W/Hz}$)  
  \end{tabular}       & $10^{-10}  $                                                                                                  \\ \hline
		$\tau_s$              
  & \begin{tabular}[c]{@{}c@{}} Sensing one \\data sample time (s)   
  \end{tabular}                                    & $0.5$                               
  \\ \hline
		$\tau_u$              
  & \begin{tabular}[c]{@{}c@{}} Transmitting one \\gradient element time (s)   
  \end{tabular}                                    & $0.02$                               
  \\ \hline
		$\delta_{k,c}^2$              & \begin{tabular}[c]{@{}c@{}} variance of residual \\ clutter after cancellation \end{tabular}                              & $1$
  \\ \hline
	\end{tabular}
\end{table}

\subsection{Effects of Sensing, Communication, and Computation}
For demonstrating how the Air-FEEL performance is affected by sensing, communication, and computation, Fig. \ref{Fig.conver} presents the influence of normalized AirComp distortion variance (i.e., $\delta_u^2/\eta$), the normalized sensing noise variance on each device (i.e., $\delta_s^2/P_{k,s}$), and the batch size of each device (i.e., $b_k$) on the testing accuracy and convergence behavior. Specifically, in Fig.  \ref{Fig.aircomploss}, we set $b_k = 350,\; \forall k$ and $\delta_s^2/P_{k,s} = 0,\; \forall k$. 
It is observed that a higher testing accuracy and a faster convergence rate are achieved when the normalized AirComp distortion variance is low. Particularly, with larger normalized AirComp distortion (e.g., $\delta_u^2/\eta\geq1e-3$), the model training oscillates and may even not converge. In Fig. \ref{Fig.sensloss}, we set $b_k = 350,\; \forall k$ and $\delta_u^2/\eta = 0$. Similarly, low normalized sensing noise variance leads to better learning performance in terms of testing accuracy and convergence rate.
In Fig. \ref{Fig.batchloss}, we set  $\delta_s^2/P_{k,s} = 0,\; \forall k$ and $\delta_u^2/\eta = 0$. It shows that a larger training batch size in each round can enhance the testing accuracy as well as the convergence rate.


The observations above verify the convergence analysis in Section \ref{Sect:Convergence}, i.e., the learning performance is affected by the normalized AirComp distortion variance, the normalized sensing noise variance on each device, and the batch size of each device.



\begin{figure}[htbp]
    \begin{center}
    \subfigure[Testing accuracy v.s. Aircomp noise]
    {
    \label{Fig.aircomploss}
    \includegraphics[width=0.45\textwidth]{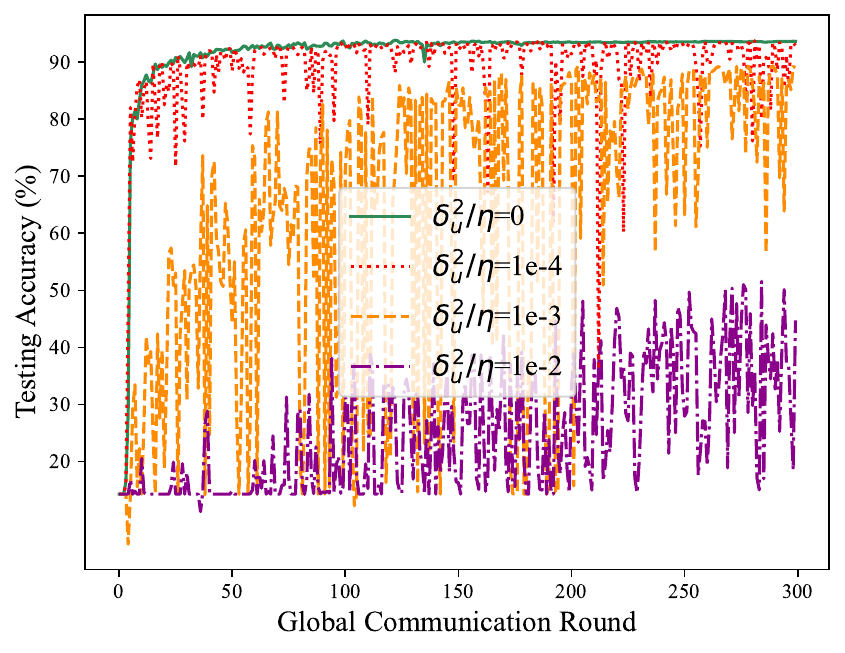}
    }
    
    \subfigure[Testing accuracy v.s. sensing noise]
    {
    \label{Fig.sensloss}
    \includegraphics[width=0.45\textwidth]{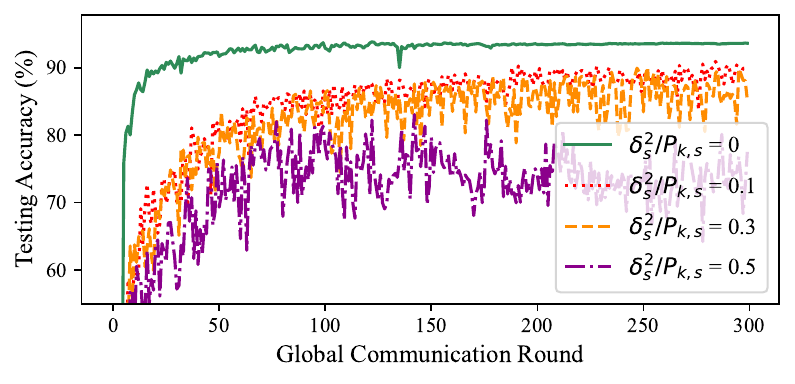}
    }

    \subfigure[Testing accuracy v.s. batch size]
    {
    \label{Fig.batchloss}
    \includegraphics[width=0.45\textwidth]{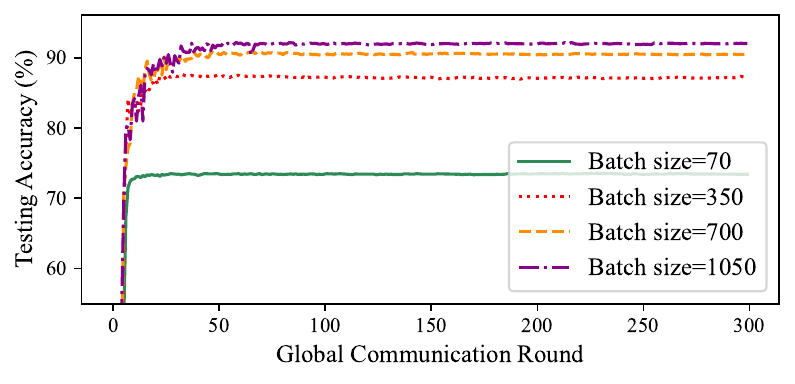}
    }    
    \caption{Effects of network resources on convergence.}%
    \label{Fig.conver}
    \end{center}
\end{figure}

\subsection{Effects of Network Resources}
To verify the superiority of the proposed ISCC scheme, the following benchmark schemes are compared under different network resources. \textcolor{black}{All these benchmark schemes demonstrate single-dimensional optimizations focusing on isolated subsystems, lacking the proposed framework's tripartite synergy across sensing, computation, and communication processes. We also compare the proposed ISCC scheme with other ISCC works, e.g., the scheme with OMA in \cite{liu2022toward}. }
\begin{itemize}
    \item \textbf{Fixed power design:} We fix the sensing power as the maximum sensing power, i.e., $P_{k,s} = P_k^{max}, \forall k$, and then optimize the remaining variables the same as the proposed one. 
    \item \textbf{Fixed computation frequency design:} We fix the computation frequency as the maximum computation frequency, i.e., $f_{k} = f_k^{max}, \forall k$, and then optimize the remaining variables the same as the proposed one.
    \item \textbf{Fixed magnitude design:} We fix the receive signal magnitude as $\eta = 0.01$, and then optimize the remaining variables the same as the proposed one.
    \item \textbf{Fixed batch size design:} We merely optimize the sub-problem of resource allocation and fix the batch size as the strictest requirements, i.e., $b_k=250, \forall k$, in per-round latency simulations and $b_k=500, \forall k$, in energy simulations.
\end{itemize}    

\begin{figure*}[htbp]
    \begin{center}
    \subfigure[Testing accuracy v.s. delay threshold]
    {
    \label{Fig.acct}
    \includegraphics[width=0.41\textwidth]{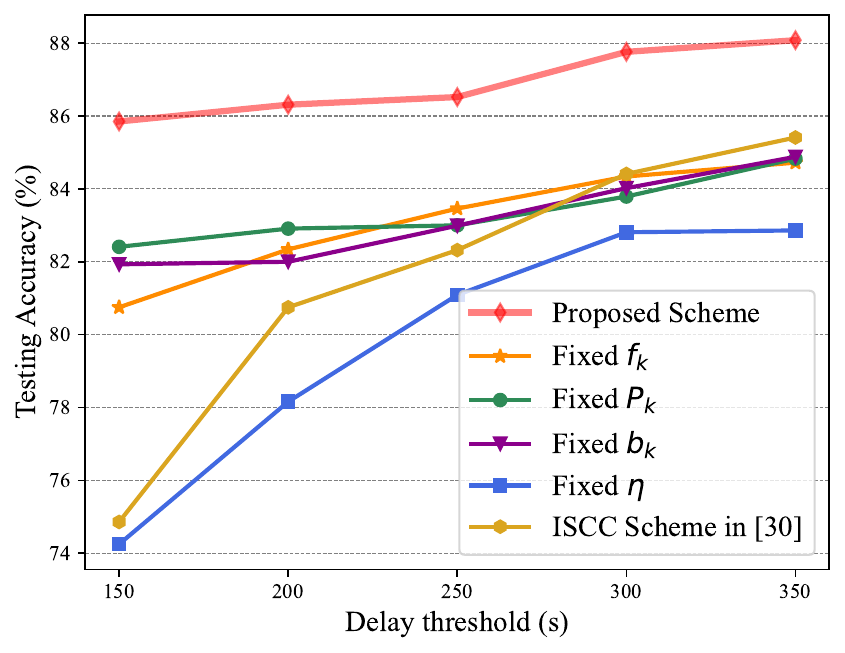}
    }
    \subfigure[Testing accuracy v.s. energy threshold]
    {
    \label{Fig.acce}
    \includegraphics[width=0.41\textwidth]{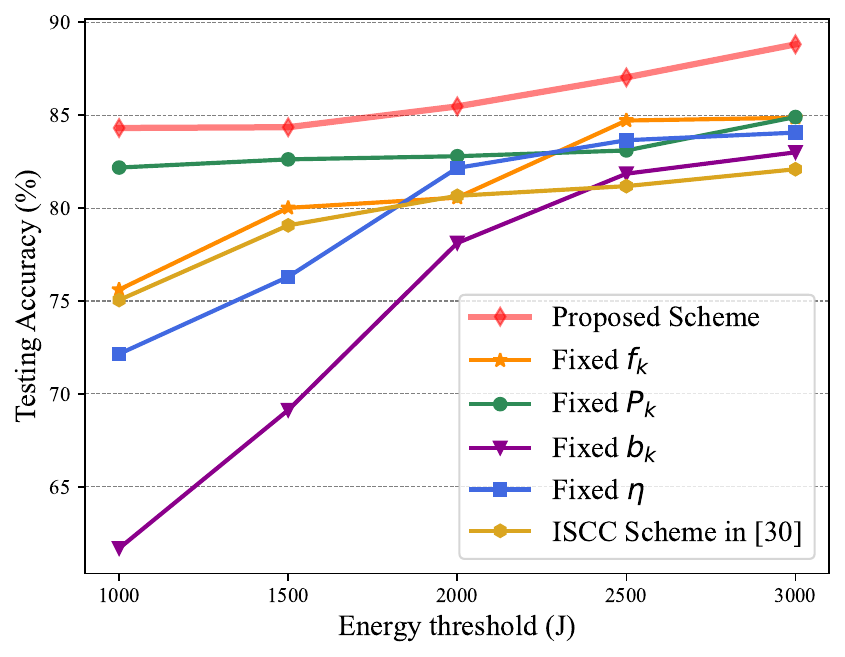}
    }
    \caption{Effects of network resources on performance.}%
    \label{Fig.acc}
    \end{center}
\end{figure*}

Fig. \ref{Fig.acct} shows the testing accuracy w.r.t. different delay thresholds. It is observed that, given a per-round latency threshold requirement, our proposed scheme outperforms other benchmark schemes and achieves the highest testing accuracy when the requirement is met. \textcolor{black}{Particularly, compared with \cite{liu2022toward}, a distinctive strength of the proposed scheme lies in consideration of sensing distortion and transmission method to obtain more sensed samples.} It is also observed that the performance tends to be better as the per-round training delay becomes larger. This is because more samples or higher sensing and AirComp SNRs can be sensed with a looser constraint. However, when the delay is large enough, the performance tends to converge. The reason is that, in such a case, energy consumption acts as a main system bottleneck.

Fig. \ref{Fig.acce} shows the testing accuracy w.r.t. different energy consumption thresholds. It is observed that our proposed scheme outperforms other benchmark schemes, which validates efficient learning. Moreover, the performance also tends to be better, with a larger energy consumption threshold. Similar to the previous case, the performance also tends to converge because the per-round training delay is the main system bottleneck.

\subsection{Effects of Device Numbers}
\textcolor{black}{
\begin{figure}
    \centering
    \includegraphics[width=0.9\linewidth]{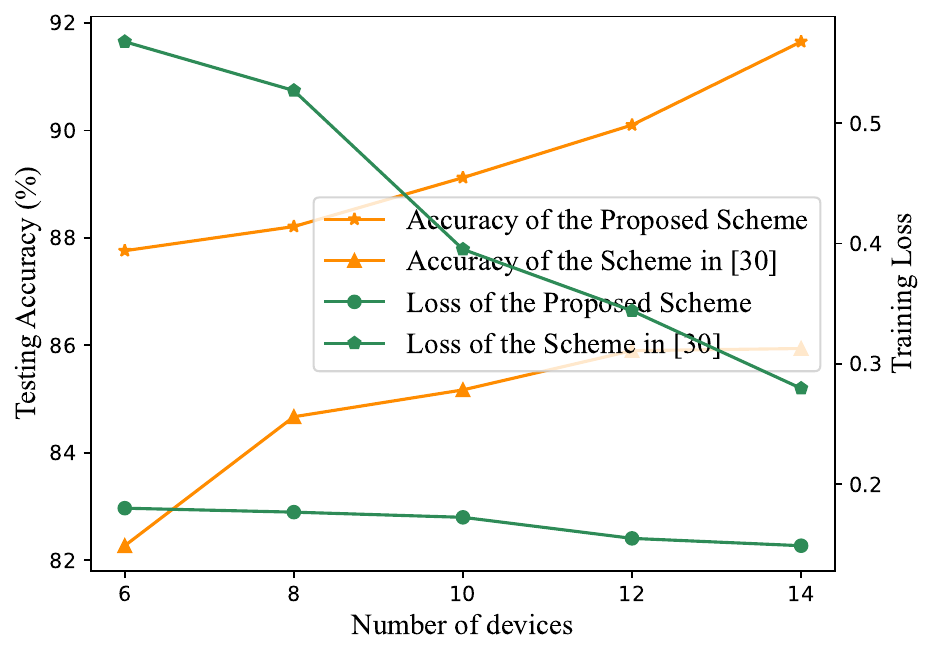}
    \caption{Effects of device numbers.}
    \label{fig:dev_num}
\end{figure}
}
\textcolor{black}{Fig. \ref{fig:dev_num} shows the testing accuracy and training loss w.r.t. the number of devices for both the proposed scheme and the scheme in \cite{liu2022toward}. As the number of devices increases, the testing accuracy improves. This is because a larger number of devices provides more sensing samples, enhancing the model's generalization capability and consequently boosting performance. Additionally, the proposed algorithm consistently outperforms the method in \cite{liu2022toward} as the number of devices increases. }

\section{Conclusion}
This paper investigated a new ISCC-based over-the-air FEEL framework, where the sensing, computation, and communication on devices are tightly coupled under the goal of enhancing learning performance. The detailed convergence analysis was first conducted, which is shown to be dependent on sensing noise, training batch size, and AirComp distortion. Particularly, we made the first attempt to characterize the effect of sensing noise on the FEEL convergence performance. Based on the theoretical convergence analysis, an ISCC scheme was proposed for enhancing the learning performance, that maximizes the loss function degradation in each round via joint batch size control and network resource allocation. The convergence behavior and the performance gain of the proposed ISCC scheme were verified via extensive simulation experiments. 
\textcolor{black}{This work opens several research directions for future investigation. These include further enhancing learning performance through strategies such as device scheduling based on integrated specific threshold or gradient importance and implementing asynchronous updating mechanisms.}

\vspace{-15pt}

\appendices
\section{Proof of Lemma \ref{Lem:Per-roundLoss}} \label{_proof_of_lemma_eachround}
According to  Assumption \ref{Assump:Smooth}, we have
\begin{equation}
    \begin{aligned}
        \mathcal{F}\left({\bf w}^{(t+1)}\right)-\mathcal{F}\left({\bf w}^{(t)}\right) 
        &\leq \left( {\bf w}^{(t+1)} - {\bf w}^{(t)} \right)^T {\bf g}^{(t)} \\
        &+ \dfrac{L}{2} \left\| {\bf w}^{(t+1)} - {\bf w}^{(t)} \right\|^2,
    \end{aligned}
\end{equation}
which, by substituting the global model update in \eqref{Eq:GlobalUpdate}, is derived as
\begin{equation}
\mathcal{F}\left({\bf w}^{(t+1)}\right)-\mathcal{F}\left({\bf w}^{(t)}\right) \leq -  \left( \alpha^{(t)} \tilde{\bf g}^{(t)}\right)^T {{\bf g}^{(t)}} +  \frac{L{\alpha^{(t)}}^2}{2} \left\|  \tilde{\bf g}^{(t)} \right\|^2.
\end{equation}
Taking expectations for both sides of the above inequality, it follows that 
\begin{equation}\label{equ:Exp}
    \begin{aligned}
        \mathbb{E} \left[\mathcal{F}\left({\bf w}^{(t+1)}\right)-\mathcal{F}\left({\bf w}^{(t)}\right) \right] 
        &\leq  - \alpha^{(t)}  \|{\bf g}^{(t)}  \|^2 \\
        &+ \frac{L{\alpha^{(t)}}^2}{2} \mathbb{E} \left[ \left\|  \tilde{\bf g}^{(t)} \right\|^2 \right],
    \end{aligned}
\end{equation}
since $\mathbb{E}( \tilde{\bf g}^{(t)} ) = {\bf g}^{(t)}$ according to Lemma \ref{Lma:GlobalGradientVector}. On the other hand, it holds that
\begin{equation}\label{Eq:GlobalGradientNorm}
\begin{aligned}
 & \mathbb{E} \left[ \left\|  \tilde{\bf g}^{(t)} \right\|^2 \right]  =  \mathbb{E} \left[ \left\|  \left( \tilde{\bf g}^{(t)} - {\bf g}^{(t)} \right) + {\bf g}^{(t)} \right\|^2 \right] \\
& = \mathbb{E} \left[ \left\|   \tilde{\bf g}^{(t)} - {\bf g}^{(t)} \right\|^2 \right] + 2 \mathbb{E} \left[ \left(   \tilde{\bf g}^{(t)} - {\bf g}^{(t)} \right)^T {\bf g}^{(t)} \right]  +  \left\|{\bf g}^{(t)} \right\|^2\\
& = \mathbb{E} \left[ \left\|   \tilde{\bf g}^{(t)} - {\bf g}^{(t)} \right\|^2 \right] +  \left\|{\bf g}^{(t)} \right\|^2,
\end{aligned}
\end{equation}
where the third equality holds because  $\mathbb{E}( \tilde{\bf g}^{(t)} ) = {\bf g}^{(t)}$. Then, by substituting the inequality of bounded variance in Lemma \ref{Lma:GlobalGradientVector} into \eqref{Eq:GlobalGradientNorm}, the following inequality is derived:
\begin{equation}
 \mathbb{E} \left[ \left\|  \tilde{\bf g}^{(t)} \right\|^2 \right] \leq \sum\limits_{k=1}^K \dfrac{b_k^{(t)}}{{b^{(t)}}^2} \left[ \sigma^2+ A^2(\delta_{k,c}^2+ \dfrac{ \delta_s^2}{P_{k,s}^{(t)}})\right] + \dfrac{\delta_u^2}{\eta^{(t)}} + \left\|{\bf g}^{(t)} \right\|^2.
\end{equation}
By substituting the above inequality into \eqref{equ:Exp}, it is further derived that 
\begin{equation}
\begin{aligned}
& \mathbb{E} \left[\mathcal{F}\left({\bf w}^{(t+1)}\right)-\mathcal{F}\left({\bf w}^{(t)}\right) \right] \leq  - \alpha^{(t)}  \|{\bf g}^{(t)}  \|^2 \\
&+ \dfrac{L{\alpha^{(t)}}^2}{2} \bigg\{ \dfrac{\delta_u^2}{\eta^{(t)}} + \left\|{\bf g}^{(t)} \right\|^2 +  \sum\limits_{k=1}^K \dfrac{b_k^{(t)}}{{b^{(t)}}^2} \bigg[ \sigma^2+ A^2(\delta_{k,c}^2+ \dfrac{ \delta_s^2}{P_{k,s}^{(t)}})\bigg] \bigg\}.
\end{aligned}
\end{equation}
By setting the learning rate as $\alpha^{(t)} = \frac{1}{ \sqrt{T} L}$, \eqref{Eq:Per-RoundLossReduction} follows. The lemma is thus proved.

\vspace{-15pt}
\section{Proof of Theorem \ref{theorem:converge}} \label{_proof_of_theorem:converge}
By rearranging \eqref{Eq:Per-RoundLossReduction}, we obtain 
\begin{equation} \label{equ:lemma_Per-roundLoss_re}
    \begin{aligned}
        &(\frac{1}{ \sqrt{T}L} - \frac{1}{2TL}) \|{\bf g}^{(t)}  \|^2
        \leq \mathbb{E} \left[\mathcal{F}\left({\bf w}^{(t)}\right)-\mathcal{F}\left({\bf w} ^{(t+1)}\right) \right] \\
        &+ \frac{1}{2TL} \bigg\{ \dfrac{\delta_u^2}{\eta^{(t)}} +  \sum\limits_{k=1}^K \dfrac{b_k^{(t)}}{{b^{(t)}}^2} \bigg[ \sigma^2+ A^2(\delta_{k,c}^2+ \dfrac{ \delta_s^2}{P_{k,s}^{(t)}})\bigg] \bigg\}.
    \end{aligned}
\end{equation}
Since $\sqrt{T} \leq T,\; T=1,2,\cdots$, the left side of \eqref{equ:lemma_Per-roundLoss_re} is scaled as 
$$\left(\frac{1}{\sqrt{T}L}-\frac{1}{2\sqrt{T}L} \right) \|{\bf g}^{(t)}  \|^2 = \frac{1}{2\sqrt{T}L} \|{\bf g}^{(t)}  \|^2.$$
Then, by multiplying $2\sqrt{T}L$ on both sides, we obtain 
\begin{equation}
    \begin{aligned}
       &\|{\bf g}^{(t)}  \|^2 \leq 2\sqrt{T}L \left[\mathcal{F}\left({\bf w}^{(t)}\right)-\mathcal{F}\left({\bf w}^{(t+1)}\right) \right] \\
       &+ \frac{1}{\sqrt{T}}  \bigg\{ \dfrac{\delta_u^2}{\eta^{(t)}} +  \sum\limits_{k=1}^K \dfrac{b_k^{(t)}}{{b^{(t)}}^2} \bigg[ \sigma^2+ A^2(\delta_{k,c}^2+ \dfrac{ \delta_s^2}{P_{k,s}^{(t)}})\bigg] \bigg\}.
    \end{aligned}
\end{equation}
Next, averaging both sides across all communication rounds $t=0, 1, \cdots, T-1$ obtains
\begin{equation}
    \begin{aligned}
        &\frac{1}{T}\sum_{t=0}^{T-1} \|{\bf g}^{(t)}  \|^2 
        \leq \dfrac{2L}{\sqrt{T}} \left[\mathcal{F}\left( {\bf w}^{(0)} \right)-\mathcal{F}\left({\bf w}^{(T)}\right) \right] \\
        &+ \frac{1}{\sqrt{T} \cdot T} \sum_{t=0}^{T-1}  \bigg\{ \dfrac{\delta_u^2}{\eta^{(t)}} +  \sum\limits_{k=1}^K \dfrac{b_k^{(t)}}{{b^{(t)}}^2} \bigg[ \sigma^2+ A^2(\delta_{k,c}^2+ \dfrac{ \delta_s^2}{P_{k,s}^{(t)}})\bigg] \bigg\}\\
        & \leq \dfrac{2L}{\sqrt{T}} \left[\mathcal{F}\left( {\bf w}^{(0)} \right)-\mathcal{F}^* \right] \\
        &+ \frac{1}{\sqrt{T} \cdot T} \sum_{t=0}^{T-1}  \bigg\{ \dfrac{\delta_u^2}{\eta^{(t)}}+  \sum\limits_{k=1}^K \dfrac{b_k^{(t)}}{{b^{(t)}}^2} \bigg[ \sigma^2+ A^2(\delta_{k,c}^2+ \dfrac{ \delta_s^2}{P_{k,s}^{(t)}})\bigg] \bigg\}.\\
    \end{aligned}
\end{equation}
As the number of rounds goes to infinity, we have \eqref{theorem:con}. This completes the proof.


\vspace{-15pt}
\section{Proof of Lemma \ref{Lem:P3} }\label{Apdx:LemP3}

Notice that problem $\mathcal{P}_3$ is convex and satisfies the Slater's condition, and thus strong duality holds between the primal problem and its dual problem. Therefore, we apply the KKT conditions to optimally solve problem  $\mathcal{P}_3$. 
Then the Lagrangian of problem $\mathcal{P}_3$ is
\begin{equation*}
\begin{aligned}
\mathcal{L}_3 =& \sum\limits_{k=1}^K \dfrac{\alpha_k^2}{ b_k } \bigg[ \sigma^2+ A^2(\delta_{k,c}^2+ \dfrac{ \delta_s^2 }{ P_{k,s} })\bigg] + \mu \left( \sum\limits_{k=1}^K \alpha_k - 1 \right)\\
+& \sum\limits_{k=1}^K \gamma_k \left( b_k \tau_s + \dfrac{b_k C}{f_k} + \left\lceil  \dfrac{N}{M} \right\rceil \tau_u - \mathcal{T} \right)\\
+& \sum\limits_{k=1}^K \lambda_k \left(  P_{k,s} b_k \tau_s + \Omega_k  b_k C f_k^2 + \dfrac{\eta \alpha_k^2 N \tau_u }{H_k}  - E_k \right),
\end{aligned}
\end{equation*}
where $\mu$, $\{\gamma_k \geq 0\}$, and $\{\lambda_k \geq 0\}$ are the dual variables associated with the constraints in \eqref{equ:P_2alpha}, \eqref{equ:P_2time}, and \eqref{equ:P_2energy}, respectively. Based on the KKT conditions with any given $\mu$, $\{\gamma_k \geq 0\}$, and $\{\lambda_k \geq 0\}$, if follows that 
\begin{align}
  &\frac{ \partial \mathcal{L}_3  }{ \partial b_k } = - \dfrac{\alpha_k^2}{ b_k^2 } \bigg[ \sigma^2+ A^2(\delta_{k,c}^2+ \dfrac{ \delta_s^2 }{ P_{k,s} })\bigg] +  \gamma_k (\tau_s + \dfrac{C}{f_k} ) \notag\\
  &~~~~~~~~~+\lambda_k (   P_{k,s} \tau_s + \Omega_k  C f_k^2 )=0, \forall k,\label{App3:KKT1}  \\
  &\frac{ \partial \mathcal{L}_3  }{ \partial \alpha_k } \!=\!\dfrac{2\alpha_k}{ b_k }\! \bigg[ \!\sigma^2\!+\! A^2(\delta_{k,c}^2\!+ \!\dfrac{ \delta_s^2 }{ P_{k,s} }\!)\!\bigg]\! +\! \mu \!+ \!\dfrac{ 2\lambda_k \alpha_k \eta N \tau_u}{H_k}\! = \!0, \forall k, \label{App3:KKT2}  
\end{align}
where \eqref{App3:KKT1} and \eqref{App3:KKT2} are the first-order derivatives of $\mathcal{L}_3  $ w.r.t. $\alpha_k$ and $b_k$, respectively. 
By combining \eqref{App3:KKT1} and \eqref{App3:KKT2} and with some manipulations, we have 
\begin{align}
\alpha_k &=  \dfrac{ (- \mu - 2\sqrt{  J_k M_k    })H_k }{  2\lambda_k \eta  N \tau_u }, \; \forall k,\label{Eq:ref1:alpha}\\
b_k &= \dfrac{H_k}{ 2\lambda_k \eta  N \tau_u } \bigg( -\mu \sqrt{ \dfrac{ M_k }{ J_k } }  - 2M_k  \bigg), \; \forall k,\label{Eq:ref1:b}
\end{align}
with $J_k = \gamma_k(\tau_s + C/ f_k ) +\lambda_k (   P_{k,s} \tau_s + \Omega_k  C f_k^2 )$ and $M_k =  \sigma^2+ A^2(\delta_{k,c}^2+  \delta_s^2 / P_{k,s} )$. Thus it remains to ﬁnd the optimal dual variables $\mu$, $\{\gamma_k \geq 0\}$, and $\{\lambda_k \geq 0\}$. 
By leveraging the primal-dual method, we update both the primal variables and the dual variables to approximately satisfy the optimality conditions. As a result, we can obtain the optimal solution for both the primal and dual variables.
At each iteration $i$, dual variables are updated as 
\begin{align}
\mu^{(i+1)}&=  \mu^{(i)} + \delta_{\mu} \frac{\partial \mathcal{L}_3}{\partial \mu^{(i)}}, \\
 \gamma_k^{(i+1)}&= \max \left\{ \gamma_k^{(i)} +  \delta_{\gamma_k}\frac{\partial \mathcal{L}_3}{\partial \gamma_k^{(i)}},\; 0 \right\}, \forall k,\\
 \lambda_k^{(i+1)}&= \max \left\{ \lambda_k^{(i)} +  \delta_{\lambda_k}\frac{\partial \mathcal{L}_3}{\partial \lambda_k^{(i)}},\; 0 \right\}, \forall k,
\end{align}
where $\delta_{\mu}$, $\{\delta_{\gamma_k}\}$, and $\{ \delta_{\lambda_k} \}$ are the step sizes.
It gradually increases the value of the dual variable until convergence, and then obtains the primary optimal solution through optimal dual variables, as shown in Lemma \ref{Lem:P3}. This lemma is thus proved.

\vspace{-15pt}
\section{Proof of Lemma \ref{Lem:P4}}\label{Apdx:LemP4}

As problem $\mathcal{P}4$ is convex and satisfies the Slater's condition, the strong duality holds between the primal problem and its dual problem. The KKT conditions are thus applied to solve problem $\mathcal{P}_4$, where the Lagrangian function is 
\begin{equation}
\begin{aligned}
& \mathcal{L}_4 = \dfrac{\delta_u^2}{\eta} +  \sum\limits_{k=1}^K \dfrac{\alpha_k^2}{ b_k } \bigg[ \sigma^2+ A^2(\delta_{k,c}^2+ \dfrac{ \delta_s^2 }{ P_{k,s} })\bigg] \\
& + \sum\limits_{k=1}^K \phi_k \left(  b_k \tau_s + \dfrac{b_k C}{f_k} + \left\lceil  \dfrac{N}{M} \right\rceil \tau_u - \mathcal{T} \right) \\
& + \sum\limits_{k=1}^K \psi_k \left(  P_{k,s} b_k \tau_s + \Omega_k  b_k C f_k^2 +  \dfrac{\eta \alpha_k^2 N \tau_u}{H_k}   - E_k \right),
\end{aligned}
\end{equation}
with $\{\phi_k \geq 0\}$ and $\{\psi_k \geq 0\}$ being multipliers. 
Hence, it holds that 
\begin{align}
&\dfrac{\partial \mathcal{L}_4 }{\partial P_{k,s} } = -\dfrac{ \alpha_k^2 A^2 \delta_s^2 }{ b_k P_{k,s}^2 } +  \psi_k b_k \tau_s = 0,\; \forall k, \label{App4:KKT1} \\
&\dfrac{\partial \mathcal{L}_4 }{\partial f_k } = - \dfrac{ \phi_kb_k C }{ f_k^2 } +  2 \psi_k \Omega_k  b_k C f_k = 0, \forall k, \label{App4:KKT2} \\
&\dfrac{\partial \mathcal{L}_4 }{\partial \eta} = -\dfrac{\delta_u^2}{\eta^2} + \sum\limits_{k=1}^K  \dfrac{\psi_k \alpha_k^2 N \tau_u}{H_k} = 0, \label{App4:KKT5} \\
& \phi_k \left(  b_k \tau_s + \dfrac{b_k C}{f_k} + \left\lceil  \dfrac{N}{M} \right\rceil \tau_u - \mathcal{T} \right)=0, \forall k, \label{App4:KKT3} \\
&\psi_k\! \left( \! P_{k,s} b_k \tau_s \!+\! \Omega_k  b_k C f_k^2 \!+\!  \dfrac{\eta \alpha_k^2 N \tau_u}{H_k}   - E_k \right)=0,\! \forall k, \label{App4:KKT4}
\end{align}
in which \eqref{App4:KKT1}, \eqref{App4:KKT2}, and \eqref{App4:KKT5} are the first-order derivative conditions of $\mathcal{L}_4$ w.r.t. $ P_{k,s}$, $f_k$, and $\eta$, respectively, \eqref{App4:KKT3} and \eqref{App4:KKT4} denote the complementary slackness conditions. 

By combining \eqref{App4:KKT1} and \eqref{App4:KKT3}, it holds that
\begin{align}
    &P_{k,s} =\dfrac{ \alpha_k A \delta_s}{ b_k \sqrt{ \psi_k\tau_s} } = \dfrac{ A \delta_s}{ b \sqrt{ \psi_k \tau_s} }, ~\forall k,
\end{align}
with second equation holding with $\alpha_k = b_k/b$. 
Furthermore, from \eqref{App4:KKT5}, we have 
\begin{align}
  \eta =\sqrt{ \delta_u^2 \bigg( \sum\limits_{k=1}^K\frac{\psi_k\alpha_k^2 N \tau_u  }{  H_k  } \bigg)^{-1} },
\end{align}
Since $\psi_k \neq 0, \; \forall k$, it holds $\phi_k \neq 0,\; \forall k$, and combining \eqref{App4:KKT3}, it follows that
\begin{equation}
 b_k \tau_s + \dfrac{b_k C}{f_k} + \left\lceil  \dfrac{N}{M} \right\rceil \tau_u = \mathcal{T},
 \label{App4:Tau}
\end{equation}
from which, we accordingly obtain
\begin{align}
f_k = \dfrac{ b_k C}{ {\mathcal{T}-\left\lceil  \dfrac{N}{M} \right\rceil \tau_u-b_k \tau_s } }, \; \forall k.
\end{align}

Then, it remains to ﬁnd the optimal multipliers $\{\phi_k \geq 0\}$ and $\{\psi_k \geq 0\}$.
We apply the primal-dual method to iteratively update both the primal variables and dual variables.
At each iteration $i$, dual variables are updated as 
\begin{align}
    \phi_k^{(i+1)}= \max \left\{ \phi_k^{(i)} + \delta_{\phi_k} \frac{\partial \mathcal{L}_3}{\partial \phi_k^{(i)}},0 \right\},  \forall k,\\
    \psi_k^{(i+1)}= \max \left\{ \psi_k^{(i)} +  \delta_{\psi_k} \frac{\partial \mathcal{L}_3}{\partial \psi_k^{(i)}},\; 0 \right\}, \forall k,
\end{align}
where $\{\delta_{\phi_k}\}$ and $\{\delta_{\psi_k}\}$ are step sizes.
By increasing the value of the dual variable until convergence, it obtains the primary solution through optimal dual solution, as shown in Lemma \ref{Lem:P4}. Hence, this lemma is completely proved.

\vspace{-15pt}
\bibliography{reference}

\end{document}